%% file: onetrace.tex
\documentclass[11pt]{article}
\usepackage{liyang}
 
\usepackage{tikz}
\usepackage{verbatim}
\usepackage{cancel}
\usepackage{setspace}
\usepackage{algorithm2e}
\usepackage{algpseudocode}

\makeatletter
\let\amsmath@bigm\bigm
\renewcommand{\bigm}[1]{%
  \ifcsname fenced@\string#1\endcsname
    \expandafter\@firstoftwo
  \else
    \expandafter\@secondoftwo
  \fi
  {\expandafter\amsmath@bigm\csname fenced@\string#1\endcsname}%
  {\amsmath@bigm#1}%
}

\newcommand{\DeclareFence}[2]{\@namedef{fenced@\string#1}{#2}}
\makeatother

\DeclareFence{\mid}{|}

\newcommand{\rnote}[1]{\footnote{{\bf \color{red}Rocco}: {#1}}}

\usepackage{todonotes}

\makeatletter
\newtheorem*{rep@theorem}{\rep@title}
\newcommand{\newreptheorem}[2]{
\newenvironment{rep#1}[1]{
 \def\rep@title{#2 \ref{##1}}
 \begin{rep@theorem}\itshape}
 {\end{rep@theorem}}}
\makeatother
\theoremstyle{plain}



\newcommand{\ignore}[1]{}

\def\colorful{1}

\ifnum\colorful=1

\newcommand{\blue}[1]{{{\color{blue}#1}}}
\newcommand{\red}[1]{{\color{brown} {#1}}}

\newcommand{\gray}[1]{{\color{gray}{#1}}}

\fi
\ifnum\colorful=0

\newcommand{\blue}[1]{{{#1}}}
\newcommand{\red}[1]{{{#1}}}

\newcommand{\gray}[1]{{{#1}}}

\fi

\usepackage{boxedminipage}

\newreptheorem{theorem}{Theorem}
\newtheorem*{theorem*}{Theorem}
\newreptheorem{lemma}{Lemma}
\newreptheorem{proposition}{Proposition}
\newreptheorem{corollary}{Corollary}
\newtheorem*{noclaim*}{Claim}

\newcommand{\Bin}{\mathrm{Bin}}

\def\Geo{\mathsf{Geometric}}

\newcommand{\Del}{\mathrm{Del}}

\newcommand{\score}{\mathsf{score}}

\usepackage{dsfont}

\renewcommand{\N}{\mathds{N}}

\newcommand{\worst}{\mathrm{worst}}

\def\LCS{\mathsf{LCS}}
\def\AvgLCS{\mathsf{AvgLCS}}

\def\deck{\mathsf{D}}

\begin{document}

\title{
Approximate Trace Reconstruction from a Single Trace
}

\author{
Xi Chen\thanks{Supported by NSF grants CCF-1703925 and IIS-1838154.} 
\\
Columbia University\\
xichen@cs.columbia.edu
\and
Anindya De\thanks{Supported by NSF grants CCF-1926872 and CCF-1910534.}
\\
University of Pennsylvania\\
anindyad@cis.upenn.edu
\and
Chin Ho Lee\thanks{Supported by the Croucher Foundation, the Simons Collaboration on Algorithms and Geometry, NSF grant CCF-1763299, and Madhu Sudan's and Salil Vadhan's Simons Investigator Awards.  Part of this work was done at Columbia University.}
\\
Harvard University\\
chlee@seas.harvard.edu
\and
Rocco A. Servedio\thanks{Supported by NSF grants CCF-1814873, IIS-1838154, CCF-1563155, and by the Simons Collaboration on Algorithms and Geometry.} 
\\
Columbia University\\
rocco@cs.columbia.edu
\and
Sandip Sinha\thanks{Supported by NSF grants  CCF-1714818, CCF-1822809, IIS-1838154, CCF-1617955, CCF-1740833, and by the Simons Collaboration on Algorithms and Geometry.}
\\
Columbia University\\
sandip@cs.columbia.edu
}

\date{}
\maketitle

\thispagestyle{empty}

\input{abstract}

\newpage

\setcounter{page}{1}

\input{intro}

\input{preliminaries}

\input{worst-case-small-rho}

\input{worst-case-small-rho-upper-bound}

\input{general-k-mixture2}



\input{worst-case-small-delta3}

\input{average-case-small-rho}

\input{average-case-small-delta}
\begin{flushleft}
\bibliography{allrefs}{}
\bibliographystyle{alpha}
\end{flushleft}

\appendix

\input{appendix}

\end{document}

%% file: abstract.tex

\begin{abstract} 

The well-known \emph{trace reconstruction problem} is the problem of inferring an unknown source string $x \in \zo^n$ from independent ``traces'', i.e.~copies of $x$ that have been corrupted by a $\delta$-deletion channel which independently deletes each bit of $x$ with probability $\delta$ and concatenates the surviving bits. 
The current paper  considers the extreme data-limited regime in which only a single trace is provided to the reconstruction algorithm. In this setting exact reconstruction is of course impossible, and the question is to what accuracy the source string $x$ can be approximately reconstructed.

We give a detailed study of this question, providing algorithms and lower bounds for the high, intermediate, and low deletion rate regimes in both the worst-case ($x$ is arbitrary) and average-case  ($x$ is drawn uniformly from $\zo^n$) models.  In several cases the lower bounds we establish are matched by computationally efficient algorithms that we provide. 

We highlight our results for the high deletion rate regime: roughly speaking, they show that 

\begin{itemize}

\item Having access to\ignore{even a small number of ``bits of trace''} a single trace is already quite useful for worst-case trace reconstruction: an efficient algorithm can perform much more accurate reconstruction, given one trace that is even only a few bits long, than it could given no traces at all.  But in contrast,

\item in the average-case setting, having access to a\ignore{ small number of ``bits of trace''} single trace is provably not very useful: no algorithm, computationally efficient or otherwise, can achieve significantly higher accuracy given one trace that is $o(n)$ bits long than it could with no traces.

\end{itemize}


\ignore{In contrast, at high deletion rates average-case and worst-case trace reconstruction provably have very different qualitative behaviors from each other.\rnote{These last two sentences are clumsily trying to say that
\begin{itemize}
\item when $\delta$ is small basically $1-\Theta(\delta)$ is achievable and best possible for both average-case and worst case; but
\item when $\rho$ is small but nonzero, in the average case setting you can't improve on a trivial (zero-trace) result, while you can improve on a trivial zero-trace performance in the worst-case setting.
\end{itemize}
}
}

\end{abstract}

%% file: intro.tex

\section{Introduction} \label{sec:intro}

The \emph{trace reconstruction problem} \cite{Kalashnik73,Lev01a,Lev01b,BKKM04} is one of the oldest and most basic algorithmic problems involving the deletion channel.  
In this problem the goal of the reconstruction algorithm is to infer an unknown $n$-bit source string $x \in \zo^n$ given access to a source of independent ``traces'' of $x$, where a trace of $x$ is a draw from $\Del_\delta(x)$.  Here $\Del_\delta(\cdot)$ is the ``deletion channel,'' which independently deletes each bit of $x$ with probability $\delta$ and outputs the concatenation of the surviving bits. The goal of the reconstruction algorithm is to correctly reconstruct the source string $x$ using as few traces and as little computation time as possible.

A surge of recent work \cite{MPV14,DOS17,NazarovPeres17,PeresZhai17,HPP18,HHP18,BCFSSfocs19,BCSSrandom19,Chase19,KMMP19,HPPZ20,Chase20,NarayananRen20,CDLSS20smoove,CDLSS20lowdeletion,ChasePeres21,CDLSS22} has addressed many different aspects and variants of the trace reconstruction problem.  The version described above corresponds to a ``worst-case'' setting, since the $n$-bit source string can be completely arbitrary; despite intensive research \cite{HMPW08,DOS17,NazarovPeres17,Chase20}, the best algorithm known for this problem, for constant deletion rate $\delta$, requires $\exp(\tilde{O}(n^{1/5}))$ traces.  Many papers such as \cite{MPV14,PeresZhai17,HPP18,BCSSrandom19,HPPZ20,Chase19,CDLSS20smoove} have also considered an ``average-case'' version of the problem in which the source string $\bx \sim \zo^n$ is assumed to be a uniform random $n$-bit string; this average-case problem is known to be significantly easier than the worst-case problem (we refer the reader to \cite{HPPZ20,Chase19} for state-of-the-art algorithmic results and lower bounds on average-case trace reconstruction at constant deletion rates).  Other problem variants which have been studied include ``population recovery'' versions in which there is a distribution over source strings rather than a single unknown source string \cite{BCFSSfocs19,BCSSrandom19,NarayananRen20}; the ``low deletion rate'' ($\delta = o_n(1)$) and ``high deletion rate'' ($\delta = 1-o_n(1)$) settings \cite{BKKM04,HMPW08,MPV14,BCFSSfocs19,NarayananRen20}; and \emph{approximate} trace reconstruction \cite{SDDF18,DRRS20,SB21,GSZ20, ChasePeres21,CDK21,CDLSS22}, in which the goal is only to obtain an approximate rather than an exact reconstruction of the unknown source string $x$, and which is the focus of the current work.

\medskip
\noindent {\bf Prior work on approximate reconstruction from few traces.}
The best algorithms known for even the easiest versions of exact trace reconstruction, such as the $\delta=O(1/\log n)$, average-case problem setting considered by \cite{BKKM04}, typically require a number of traces that grows with $n$ to achieve exact reconstruction.\footnote{Indeed, at deletion rate $\delta=O(1/\log n)$, it is easy to see that given a sample of $o({\frac {\log n}{\log \log n}})$ traces, with high probability there will be coordinates of the source string that are deleted from all of the traces in the sample.}  An attractive feature of the recent works \cite{ChasePeres21,CDLSS22} is that they give provable performance guarantees for \emph{approximate} trace reconstruction even when only \emph{constantly} many traces are available. In more detail, \cite{CDLSS22} gave near-matching upper and lower bounds on the best possible reconstruction accuracy that any algorithm can achieve given $M=O(1/\delta)$ traces from $\Del_\delta(\bx)$ in the average-case setting.  \cite{ChasePeres21} showed that for any constants $\delta,\eps$, there is some constant $M=M(\eps,\delta)$ such that an $M$-trace algorithm can achieve reconstruction error at most $\eps$ given traces from $\Del_\delta(\bx)$ in the average-case setting.\footnote{Several other recent papers \cite{SDDF18,DRRS20,GSZ20,CDK21,SB21} have also studied approximate trace reconstruction, but focusing on different aspects that make them less relevant to the present paper;  see \cite{CDLSS22} for a detailed discussion of those works.}

Results such as \cite{ChasePeres21,CDLSS22}, which shed light on what can be achieved given constantly many traces, can be particularly valuable in settings where only a severely limited number of traces are available and the goal is to do as well as possible with the data at hand. Such settings motivate the present paper, which, as we now describe, studies trace reconstruction in the ultimate data-constrained regime.

\medskip
\noindent {\bf This work: Approximate reconstruction from a single trace.}  We consider the problem of recovering an unknown $n$-bit source string $x$ as accurately as possible given only \emph{a single trace} from $\Del_\delta(x)$.  Despite the simplicity and naturalness of this problem, it does not seem to have been considered in prior work.

We give a detailed study of this problem, analyzing both the worst-case setting of an arbitrary unknown source string $x$ as well as the average-case setting of a uniform random $\bx \sim \zo^n$.  In each of these settings we consider both the low ($\delta = o_n(1)$), medium ($\delta = \Theta(1)$), and high ($\delta=1-o_n(1)$) deletion rate regimes.
In a number of cases we give upper bounds on the approximate reconstruction accuracy that any one-trace algorithm can achieve, which are essentially matched by corresponding one-trace algorithms that we provide. (All of the algorithms we give are computationally efficient.)  
For some problem variants our upper bounds on the best achievable accuracy extend beyond one-trace algorithms to algorithms that receive multiple traces.

We view our results as a first investigation of one-trace reconstruction, and reiterate that very little was previously known for any of the problem variants that we consider. We describe the state of prior knowledge in the context of different specific problem variants when describe our results in \Cref{sec:our-results} below.

\subsection{Our results} 
\label{sec:our-results}

We are interested in the abilities and limitations of algorithms $A$ which receive as input a single trace $y$ from an unknown $n$-bit source string $x$ and which output an $n$-bit hypothesis string $\wh{x}=A(y).$  We measure the accuracy of $\wh{x}$ with respect to $x$ by the length of the longest common subsequence $|\LCS(x,\wh{x})|$.  LCS is closely related to edit distance, since if $|\LCS(x,\wh{x})|=n-k$ for two $n$-bit strings $x$ and $\wh{x}$, then $\wh{x}$ can be converted into $x$ by a sequence of $k$ deletions and $k$ insertions (and this is best possible).  The goal of an approximate reconstruction algorithm in this setting is to output a hypothesis string $\wh{x}$ for which the expectation\footnote{In the worst-case setting this expectation is over the random draw of the trace $\by$ from $\Del_\delta(x)$; in the average-case setting, this expectation is also over the uniform random draw of the source string $\bx \sim \zo^n$.  We give more details and a precise formulation in \Cref{sec:alg-performance-notation}.} of $|\LCS(x,\wh{x})|$ is guaranteed to be as large as possible; thus positive (algorithmic) results in our setting yield \emph{lower bounds} on how large an  expected value of $|\LCS(x,\wh{x})|$ can be achieved, while impossibility results for algorithms give \emph{upper bounds} on the best achievable expected $|\LCS(x,\wh{x})|.$ 

As alluded to earlier, we consider both the setting of a worst-case (arbitrary) $x \in \zo^n$ and the setting of a uniform random $\bx \sim \zo^n$. We note that algorithmic results (lower bounds on $\E[|\LCS(x,\wh{x})|]$ for the worst-case setting carry over to the average-case setting, while impossibility results (upper bounds on $\E[|\LCS(x,\wh{x})|]$) for the average-case setting carry over to the worst-case setting.

\Cref{sec:worst-case} presents our results for the worst-case setting and \Cref{sec:average-case} presents our results for the average-case setting. In each of these sections we first present our results (upper and lower bounds) for the high and medium deletion rate regimes, and then the low deletion rate regime. 

\subsubsection{Worst-case one-trace reconstruction} \label{sec:worst-case}

We first consider the high deletion rate regime.  It is convenient to let $\rho := 1-\delta$ denote the \emph{retention rate}, so in the high deletion rate regime we have $\rho=o(1)$.

If $\rho$ is too small (as a function of $n$) then it is easy to see that no nontrivial performance is possible.  In particular, if $\rho=o(1/n)$, then by Markov's inequality with probability $1-o(1)$ a trace $\by \sim \Del_\delta(x)$ is zero bits long, and in this case a reconstrution algorithm cannot even distinguish between the two possibilities $x=0^n$ and $x=1^n$. Consequently, if $\rho=o(1/n)$ then the largest expected LCS achievable by a one-trace algorithm is at most $(1/2 + o_n(1))n$ (and $n/2$ is trivially achieved by outputting any string with an equal number of 0's and 1's).

Our first positive result shows that --- perhaps surprisingly --- if $\rho$ is only slightly larger, then it is already possible to do much better than the above trivial bound:

\begin{theorem} [Worst-case algorithm, small retention rate, informal statement] \label{thm:worst-case-small-rho-informal}
For any $\rho = \omega(\log(n)/n)$, there is a worst-case one-trace algorithm that achieves expected LCS at least $(2/3-o(1))n$. 
Moreover, for any retention rate $\rho \geq \omega(1/n^{1/3})$, there is a worst-case one-trace algorithm that achieves expected LCS at least $(2/3 + c \rho)n$, where $c>0$ is an absolute constant.
\end{theorem}

The key to \Cref{thm:worst-case-small-rho-informal} is a (to the best of our knowledge novel) notion of an \emph{LCS-cover}, and a simple construction of an extremely small LCS-cover consisting of just two strings.  This already suffices to give the first sentence of \Cref{thm:worst-case-small-rho-informal}; the second sentence, improving the LCS bound to $(2/3 + c\rho)$, is obtained via a win-win analysis which considers whether or not the single received trace has many ``long runs''. Roughly speaking, if the trace has many long runs then this indicates that the source string $x$ is highly structured in a way (containing many long segments that are almost all-0 or almost all-1) that makes it easy to achieve a large LCS, and if the trace has few long runs then the source string $x$ must have many 01 alternations, which can be leveraged to get an LCS larger than $2n/3$.

\Cref{thm:worst-case-small-rho-informal} can be viewed as saying that having a $\log n$-bit trace already makes it possible to achieve an LCS of at least $(2/3 - o(1))n.$ Complementing \Cref{thm:worst-case-small-rho-informal}, we show that even having a $n^{0.999}$-bit trace does not make it possible to achieve an LCS of $(2/3 + c)n$ for any $c>0$:

\begin{theorem} [Worst-case upper bound on any algorithm, small retention rate, informal statement] \label{thm:worst-case-small-rho-upper-bound-informal}
Fix any $\eps>0$.  
For retention rate $\rho = 1/n^{\eps}$, no one-trace algorithm can achieve expected LCS greater than $(2/3 + o(1))n$ in the worst-case setting.
\end{theorem}

See \Cref{thm:worst-case-small-rho-upper-bound}
 for a detailed theorem statement, which extends \Cref{thm:worst-case-small-rho-upper-bound-informal} to give an upper bound on the performance of algorithms that receive multiple traces. 
\Cref{thm:worst-case-small-rho-upper-bound-informal} leverages a recent deep result of Guruswami, Haeupler, and Shahrasbi~\cite{GHS20} analyzing a code due to Bukh and Ma \cite{BukhMa14}. We take advantage of the highly repetitive structure of the Bukh--Ma codewords to combine the \cite{GHS20} result with a construction of a family of distributions over Bukh--Ma codes
such that the \emph{$k$-decks}\footnote{The $k$-deck of a single string is the multiset of all length-$k$ subsequences of the string, and the $k$-deck of a distribution over strings is the corresponding mixture of $k$-decks of the constituent strings in the mixture; see \Cref{sec:deck-notation} for detailed definitions.}
 of all of the different distributions coincide.  This in turn lets us show that a single trace does not have enough information to make more accurate reconstruction than (essentially) LCS $2n/3$ possible.



\Cref{thm:worst-case-small-rho-informal} sheds light on the high deletion rate and medium deletion rate regimes of one-trace reconstruction.  Turning to the medium and low deletion rate regimes, if the retention rate $\rho$ is large enough (at least some absolute constant), then the algorithm used for \Cref{thm:worst-case-small-rho-upper-bound-informal} is no longer best possible, since it would be better to simply output any string $\wh{x}$ that contains the trace $\by$ as a subsequence. This is because, as observed in \cite{CDLSS22}, any such string $\wh{x}$ achieves expected LCS at least $\E[|\by|]=\rho n = (1-\delta)n$.

Can better performance than this naive $(1-\delta)n$-length LCS be achieved in the medium and low deletion rate regimes?  We give an improvement by constructing a hypothesis string $\wh{\bx}$ that randomly intermingles random bits with the bits of $\by$. A careful analysis of the LCS between this $\wh{\bx}$ and the source string $x$ yields the following:

\begin{theorem} [Worst-case algorithm, small deletion rate, informal statement]\label{thm:worst-case-small-delta-informal}
There is a worst-case one-trace algorithm that achieves expected LCS at least $(1 - \delta + \delta^2/2 - \delta^3/2 + \delta^4/2 - \delta^5/2 - o(1))n$ for deletion rate $\delta.$
\end{theorem}

Given \Cref{thm:worst-case-small-delta-informal}, it is natural to ask about limitations of one-trace reconstruction in the low deletion rate regime. Taking $M=1$ in the main lower bound result (Theorem~1.2) of \cite{CDLSS22}, that result shows that no one-trace algorithm can achieve expected LCS greater than $(1-\delta^C)n$ in the worst-case setting, where $C$ is some absolute (large) constant. In \Cref{sec:average-case} we will see that \Cref{thm:small-delta-average-case-upper-bound-informal} establishes a stronger and near-optimal bound even for the more challenging average-case setting.

\subsubsection{Average-case one-trace reconstruction} \label{sec:average-case}

The average-case setting of one-trace reconstruction turns out to present some unexpected challenges due to connections with difficult unresolved problems in the combinatorics of words.  To see this, let us first consider the problem of average-case trace reconstruction from \emph{zero} traces; so the reconstruction algorithm receives no input at all, and simply aims to output the $n$-bit string $\wh{x}$ which maximizes the expected value of $|\LCS(\wh{x},\bx)|$ across uniform random $\bx \sim \zo^n$.  In contrast with the worst-case setting (where no zero-trace algorithm can achieve expected LCS better than 1/2 because the source string $x$ could be chosen uniformly at random from $\{0^n, 1^n\}$), in the average-case setting the hypothesis string $\wh{x}=(01)^{n/2}$ already achieves $\E[|\LCS(\wh{x},\bx)|] \geq (3/4 - o(1))n$: first greedily match the 0's in $\wh{x}$ with the 0's in $\bx$ from left to right, and then opportunistically augment these $\approx n/2$ matching edges with edges matching pairs of 1's where possible. So nontrivial performance is possible, even with zero traces, in the average-case setting.


Can we do better with a smarter choice of the hypothesis string $\wh{x}$? A natural idea is to select $\wh{\bx}$ uniformly at random from $\zo^n$. The performance of this zero-trace algorithm is captured by the \emph{Chv\'{a}tal--Sankoff constant}
\[
\gamma_2 := \lim_{n \to \infty}{\frac {\E_{\bx,\wh{\bx} \sim \zo^n}[|\LCS(\bx,\wh{\bx})|]}{n}}
\]
(the ``2'' is because we are working with the binary alphabet); the existence of this limit is an easy consequence of the superadditivity of LCS between random strings (using Fekete's Lemma \cite{Subadditivity:wikipedia}). Despite much investigation over more than 40 years, the value of $\gamma_2$ is not known: in 1975 Chv\'{a}tal and Sankoff showed that $0.727273 \leq \gamma_2 \leq 0.866595,$ and the current state of the art bounds, due to Lueker \cite{Lueker09}, are that $0.788071 \leq \gamma_2 \leq 0.826280$ \cite{ChvatalSankoff:wikipedia}.

A superadditivity argument similarly establishes the existence of the limit 
\begin{equation}
\label{eq:c2}
c_2 := \lim_{n \to \infty} \max_{\wh{x} \in \zo^n} {\frac {\E_{\bx \sim \zo^n}[|\LCS(\bx,\wh{x})|]}{n}},
\end{equation}
which corresponds to the performance of the information-theoretic optimal zero-trace algorithm for the average-case setting. Even less is known about $c_2$ than $\gamma_2$; Bukh and Cox \cite{BC22} have shown (via an involved argument and an automated search) that $c_2 \geq 0.82118$, and we show in \Cref{ap:upper-bound-c2} that $c_2 \leq 0.88999$, but more detailed bounds on the value of $c_2$ do not seem to be known, nor is it known what strings might achieve this optimal bound \cite{Bukh22}.

Given these challenges in understanding zero-trace reconstruction in the average-case setting, the prospects of analyzing one-trace average-case reconstruction may appear dim.
Perhaps surprisingly, for the low deletion rate regime and medium deletion rate regime it turns out that the difficulty of analyzing zero-trace reconstruction is the only barrier to showing an upper bound on average-case one-trace reconstruction.  This is shown in the following theorem, which gives an upper bound on average-case one-trace reconstruction in terms of the quantity $c_2$ from \Cref{eq:c2}:

\begin{theorem} [Average-case upper bound on any algorithm, small retention rate, informal statement] \label{thm:small-rho-average-case-informal}
Let $L_{1,\avg}(\delta,n)$ denote the best expected LCS achievable by any one-trace algorithm at deletion rate $\delta$ in the average-case setting.  Then we have $c_2 \leq \lim_{n \to \infty} {\frac {L_{1,\avg}(\delta,n)}{n}} \leq c_2+ \rho.$ 
\end{theorem}

\Cref{thm:small-rho-average-case-informal} tells us that for any $\rho=o_n(1)$ retention rate, it is not possible to asymptotically improve on the performance of the best zero-trace algorithm.
In fact, in \Cref{thm:small-rho-average-case} we give a generalization of \Cref{thm:small-rho-average-case-informal} which gives an upper bound on the performance of algorithms that receive more than one trace.  The proof is based on a careful analysis, using a coupling argument, of the \emph{a posteriori} distribution of the random source string $\bx$ given the received collection of traces.

Finally, we consider upper and lower bounds which are applicable for the medium and small deletion rate regime.  In the average-case setting, the algorithm of \Cref{thm:worst-case-small-delta-informal} can be shown to have better performance than was established in \Cref{thm:worst-case-small-delta-informal} for the worst-case setting:

\begin{theorem} [Average-case algorithm, small deletion rate, informal statement] \label{thm:small-delta-average-case-algorithm-informal}
There is an average-case one-trace algorithm that achieves expected LCS at least 
$ \left( 1 - \delta + \frac12\delta^2   + \frac{17}{8}\delta^4 + \frac{55}{8}  \delta^5 -o(1) \right) n$ for deletion rate $\delta.$
\end{theorem}

Given \Cref{thm:small-delta-average-case-algorithm-informal}, it is natural to investigate the best possible performance of any one-trace algorithm in the average-case setting for small $\delta$. A relatively simple probabilistic argument (which is based on a union bound across all possible matchings, and which we give in \Cref{sec:small-delta-average-case-weak-upper-bound}) shows that the expected LCS achieved by any one-trace algorithm can be at most $(1 - \Omega(\delta/\log(1/\delta))) \cdot n$.  Via a more involved probabilistic argument we strengthen this to a $1-\Theta(\delta)$ bound:  

\begin{theorem} 
[Average-case upper bound on any algorithm, small retention rate, informal statement] 
\label{thm:small-delta-average-case-upper-bound-informal}
For any deletion rate $\delta = \omega(1/n)$, no one-trace algorithm can achieve expected LCS greater than $(1-c \delta)n$ in the average-case setting, where $c$ is some absolute constant.
\end{theorem}

We observe that by virtue of \Cref{thm:small-delta-average-case-algorithm-informal}, \Cref{thm:small-delta-average-case-upper-bound-informal} is best possible up to the hidden multiplicative constant on the $\delta$-term.

\subsection{Future work}

A natural first goal for future work is to obtain sharper results. For example, if the deletion rate $\delta$ is 0.1,
 what is the largest constant $c$ such that expected LCS $cn$ can be achieved in the worst-case setting? In the average-case setting? What if $\delta = 0.9$?  We do not currently have sharp answers to questions such as these.

A different goal is to go beyond one-trace reconstruction.  While our negative results \Cref{thm:worst-case-small-rho-upper-bound-informal} and \Cref{thm:small-rho-average-case-informal} extend to algorithms that receive multiple traces, it would be interesting to extend positive results such as \Cref{thm:worst-case-small-rho-informal}, \Cref{thm:worst-case-small-delta-informal} and \Cref{thm:small-delta-average-case-algorithm-informal} to the setting of multiple traces. In this context we mention the work of Chakraborty et al.~\cite{CDK21}, which gave an average-case algorithm for approximate trace reconstruction from three traces in an insertion-deletion model.

%% file: preliminaries.tex

\section{Preliminaries} \label{sec:preliminaries}

\medskip

\noindent {\bf Notation.}
Given a positive integer $n$, we write $[n]$ to denote $\{1,\ldots,n\}$.
Given two integers $a\le b$ we write $[a:b]$ to denote $\{a,\ldots,b\}$. 
We write $\ln$ to denote natural logarithm and $\log$ to denote logarithm to the base 2.
We denote the set of non-negative integers by $\Z_{\geq 0}$.
We  write ``$a=b\pm c$'' to indicate that $b-c\le a\le b+c$.
It will be convenient for us to index a binary string~$x \in \zo^n$
  using $[1:n]$ as~$x=(x_1,\dots,x_{n})$. 

\medskip

\noindent {\bf Distributions.}
When we use bold font such as $\bD, \by, \bz$, etc., it indicates that the entity in question is a random variable.
We write ``$\br \sim {\cal P}$'' to indicate that random variable~$\br$~is 
  distributed according to probability distribution ${\cal P}$.  If $S$ is a finite set we write ``$\br \sim S$'' to indicate that $\br$ is distributed uniformly over $S$.
 
 We write $\Geo(\rho)$ to denote the geometric distribution with parameter $\rho$, i.e.~the number of Bernoulli trials with success probability $\rho$ needed to get one success, supported on the set $\{1,2,\dots\}$.
 We will use the following tail bound for sums of independent geometric random variables:

\begin{claim} \label{claim:neg-binomial}
Let $\rho \in [0,1]$ and let $\bG_1, \ldots, \bG_m$ be $m$ independent geometric random variables with each $\bG_i \sim \Geo(\rho)$.
For any $\gamma \in [0,1]$, we have
  \[
    \Pr\left[ \abs[\bigg]{\sum_{i=1}^m \bG_i - \rho^{-1} m } \ge \gamma \rho^{-1} m \right]
    \le e^{-\Omega(\gamma^2 m) }  .
  \]
\end{claim}
\begin{proof} By coupling $(\bG_1,\dots,\bG_m)$ with a draw from the Binomial distribution $\Bin(n,\rho)$, we observe that $\sum_{i=1}^m \bG_i \ge n$ if and only if $\Bin(n, \rho) < m$.
  Let $n_h := (1+\gamma) \rho^{-1} m$ and $n_\ell := (1-\gamma) \rho^{-1} m$.
  We have
  \begin{align*}
    \MoveEqLeft
    \Pr\left[ \abs[\bigg]{\sum_{i=1}^m \bG_i - \rho^{-1} m } \ge  \gamma\rho^{-1} m \right] \\
    &= \Pr\bigl[ \Bigl( \Bin\bigl( (1 + \gamma) \rho^{-1} m, \rho \bigr) < m \Bigr) \vee \Bigl( \Bin\bigl( (1 - \gamma) \rho^{-1} m, \rho \bigr) > m \Bigr) \biggr] \\
    &= \Pr\biggl[ \Bigl( \Bin\bigl( n_h, \rho \bigr) < \frac{1}{1+\gamma} \cdot \rho n_h \Bigr) \vee \Bigl( \Bin\bigl( n_\ell, \rho \bigr) > \frac{1}{1-\gamma} \cdot \rho n_\ell \Bigr) \biggr] \\
    &= \Pr\biggl[ \Bigl( \Bin\bigl( n_h, \rho \bigr) < \left( 1 - \frac{\gamma}{1+\gamma} \right) \rho n_h \Bigr) \vee \Bigl( \Bin\bigl( n_\ell, \rho \bigr) > \left( 1 + \frac{1}{1-\gamma} \right) \rho n_\ell \Bigr) \biggr] \\
    &\le e^{-\Omega(\gamma^2 m) }, \qedhere
  \end{align*}
where the inequality is a standard Chernoff bound.\end{proof}

\medskip

\noindent {\bf Deletion channel and traces.}
Throughout this paper the parameter $0 <\delta < 1$ denotes the \emph{deletion probability}.  Given a string $x \in \zo^n$, we write $\Del_\delta(x)$ to denote the distribution of the string that results from passing  $x$ through the $\delta$-deletion channel (so the distribution $\Del_\delta(x)$ is supported on $\zo^{\leq n}$), and we refer to a string in the support of $\Del_\delta(x)$ as a \emph{trace} of $x$.  Recall that a random trace $\by \sim \Del_\delta(x)$ is obtained by independently deleting each bit of $x$ with probability $\delta$ and concatenating the surviving bits.\hspace{0.05cm}\footnote{In this work we assume that the deletion probability $\delta$ is known to the reconstruction algorithm.\ignore{ \red{Rocco: Do we need this? Maybe yes to get from \Cref{thm:main} to \Cref{thm:main2}?}}}
We may view the draw of a trace $\by$ from $\Del_\delta(x)$ as a two-step process:  first a set $\bD$ of \emph{deletion locations} is obtained by including each
  element of $[n]$ independently with probability $\delta$, and then 
  $\by$ is set to be $\smash{x_{[n]\setminus \bD}}$. 
    
\medskip

\noindent {\bf $\LCS$ and matchings.}
We write $\LCS(x,x')$ to denote the longest common subsequence between two strings $x$ and $x'$ and $|\LCS(x,x')|$ to denote its length. 
A \emph{matching} $M$ between two strings $x,x' \in \zo^\ast$ is a list of pairs $(v_1,v'_1), (v_2,v'_2),\dots$ such that $v_1 \leq v_2 \leq \cdots$, $v'_1 \leq v'_2 \leq \cdots$, and for every $t$ we have $x_{v_t} = x'_{v'_t}.$ The \emph{size} of a matching is the number of pairs.  We note that the largest matching between $x$ and $x'$ is of length 
$|\LCS(x,x')|$.

\medskip

%
%
%
%
%

\noindent {\bf An asymptotic bound on binomial coefficients.}
We recall the following standard bound on binomial coefficients:
\begin{fact} [\cite{vanLint}, Theorem~1.4.5]
\label{fact:standard-bound}
For~$0 \leq k \leq n/2$, we have $\sum_{i=0}^{k} {n \choose i} \leq 2^{H(k/n)n}$, where $H(x) = x\log(1/x) + (1-x)\log(1/(1-x))$ is the binary entropy function.
\end{fact}

\subsection{The average-case setting} \label{sec:avg}
We record the following simple observation, which is useful for analyses of the average-case setting: 
 
\begin{observation}[\emph{A posteriori} distribution of a uniform random source string given one trace] \label{ob:post1}
Let $\bx$ be a uniform random source string from $\zo^n$.
Given any fixed outcome $y \in \zo^m$ of a single trace $y =\by \sim \Del_\delta(\bx)$,
the \emph{a posteriori} distribution of $\bx$ given $y$ is as follows:

\begin{enumerate}

\item Draw a uniform random $m$-element subset $\bS \sim {[n] \choose [m]}$ of $[n]$ (say $\bS = \{\bs_1,\dots,\bs_m\}$ where $1 \leq \bs_1 < \cdots < \bs_m \leq n$);

\item For each $i \in [m]$ set $\bx_{\bs_i}=y_i$ (i.e. fill in the locations in $\bS$ from left to right with the  bits of $y$), and for each $j \notin \bS$ set $\bx_j$ to an independent uniform element of $\zo$.
\end{enumerate}
\end{observation}

We write ``$\bx \sim y$'' to indicate that $\bx$ has the distribution described above. We note that a somewhat counterintuitive corollary of \Cref{ob:post1} is the following: in the average-case setting (when $\bx$ is uniform random), even if the received trace is the string $1^m$, the \emph{a posteriori} distribution of the $n-m$ ``unseen bits'' of $\bx$ is that they are independent and uniform random.

An easy corollary of \Cref{ob:post1} is the following:

\begin{corollary}  \label{cor:uniform}  
For $\bx$ a uniform random source string from $\zo^n$,
given any fixed outcome $y \in \zo^m$ of a single trace $y =\by \sim \Del_\delta(\bx)$,
the \emph{a posteriori} distribution of the other $n-|y|$ bits $\bx_{\bD}$ of $\bx$ is that they are distributed as a uniform random element of $\zo^{[n] \setminus |y|}$.
\end{corollary}

\subsection{One-trace and few-trace algorithms.} \label{sec:alg-performance-notation}

\noindent {\bf Optimal worst-case algorithms.}
We introduce the notation
$L_{1,\worst}(\delta,n)$ to denote the largest possible $\LCS$ that can be achieved in expectation by any one-trace algorithm under deletion rate $\delta$ in the worst-case setting, i.e.
\begin{equation} \label{eq:worst-case-one-trace}
L_{1,\worst}(\delta,n) :=  \max_{A} \min_{x \in \zo^n} \Ex_{\by \sim \Del_{\delta}(x)}
[|\LCS(A(\by,n),x)|],
\end{equation}
where the maximum is taken over all algorithms $A$ that take as input the values $n, \delta$ and a single trace $\by$, and output an $n$-bit hypothesis string (denoted $A(\by,n)$ in the expression above).   We observe that (\ref{eq:worst-case-one-trace}) could be extended to allow the algorithm $A$ to be randomized (and have the expectation be also over the randomness of $A$), but we do not do this since the optimal algorithm in (\ref{eq:worst-case-one-trace}) can without loss of generality be taken to be deterministic.

We will sometimes consider the optimal performance of $t$-trace algorithms for $t>1$, so we extend the above definition in the obvious way to algorithms that are given $t$ independent traces, i.e.
\begin{equation} \label{eq:worst-case-t-traces}
L_{t,\worst}(\delta,n) :=  \max_{A} \min_{x \in \zo^n} \Ex_{\by^{(1)},\dots,\by^{(t)} \sim \Del_{\delta}(x)}
[|\LCS(A(\by^{(1)},\dots,\by^{(t)},n),x)|].
\end{equation}

\noindent {\bf Optimal average-case algorithms.}
We use similar notation to capture the optimal performance of one-trace and $t$-trace algorithms in the average-case setting:
\begin{equation} \label{eq:average-case-one-trace}
L_{1,\avg}(\delta,n) :=  \max_{A} \Ex_{\bx \sim \zo^n} \Ex_{\by \sim \Del_{\delta}(\bx)}
[|\LCS(A(\by,n),\bx)|],
\end{equation}

\begin{equation} \label{eq:avg-case-t-traces}
L_{t,\avg}(\delta,n) :=  \max_{A} \Ex_{\bx \sim \zo^n} \Ex_{\by^{(1)},\dots,\by^{(t)} \sim \Del_{\delta}(\bx)}
[|\LCS(A(\by^{(1)},\dots,\by^{(t)},n),\bx)|].
\end{equation}

%% file: worst-case-small-rho.tex

\section{Worst-case one-trace reconstruction, small retention rate} \label{sec:worst-case-small-rho}

\subsection{An efficient algorithm}

%
We prove \Cref{thm:worst-case-small-rho-informal} in this subsection. 
We start with the first part of \Cref{thm:worst-case-small-rho-informal}, i.e., 
  when the retention rate $\rho$ is large enough ($\omega(\log (n)/n)$) 
  that a nontrivial number of bits are expected to be present in a random trace, 
  then a simple computationally efficient one-trace algorithm can achieve an LCS significantly better than $n/2$.

\subsubsection{A useful structural result and a $(2/3 - o(1))$-LCS algorithm for $\rho = \omega(\log(n)/n)$}\label{sec:LCSalgo}

It is helpful for us to consider the following preliminary problem: we are not given any traces, and the goal is to output \emph{a list} of $m$-bit candidates such that the unknown source string $x \in \zo^n$ has large LCS with one of the candidate strings in our list.
This motivates the following definition:

\begin{definition}[LCS-cover] \label{def:lcs-cover}
Let $m$ and $n$ be two positive integer.
We say a set $S \subseteq \zo^m$ is an \emph{$h$-LCS cover} for strings of length $n$ if for every $x \in \zo^n$ we have 
\[
  \abs[\big]{\LCS(S, x)}
  := \max_{s \in S}\hspace{0.04cm} \abs[\big]{\LCS(s, x)}
  \ge h .
\]
\end{definition}

The following simple claim shows that when $m$ is within a factor of two of $n$, there is a (perhaps surprisingly) good LCS cover consisting of at most two strings:

\begin{claim} \label{claim:lcs-cover}
  For every $m \in [n/2, 2n]$, there exists a $( (n+m)/3)$-LCS-cover of size at most $2$.
\end{claim}

\begin{proof}
  We first consider the extreme settings of $m = 2n$ and $m=n/2$.
  When $m = 2n$, we~have $\abs{\LCS((01)^n, x)} = n$ for every $x \in \zo^n$, and thus $\{ (01)^n \}$ is an $n$-LCS cover (of size $1$).
  When $m = n/2$, every $x \in \zo^n$ either contains $n/2$ many 1s or this many 0s, and so either $\abs{\LCS(0^{n/2}, x)} \ge n/2$ or $\abs{\LCS(1^{n/2}, x)} \ge n/2$, and hence the set $\{0^{n/2}, 1^{n/2}\}$ is a $(n/2)$-LCS cover of size $2$.

We interpolate between these two cases to handle general $m$'s.
  Write $m = 2a + b$ and $n = a + 2b$ for some $a$ and $b$ (so $a = m - (n+m)/3$ and $b = n - (n+m)/3$).
  Consider
\begin{equation} \label{eq:LCS-cover}
    S := \bigl\{ (01)^a 0^b, (01)^a 1^b \bigr\} \subseteq \zo^m .
\end{equation}
  Given $x \in \zo^n$, we can write $x = x_1 \circ x_2$ where $x_1 \in \zo^a$ and $x_2 \in \zo^{2b}$, and we get that
  \[
    \abs[\big]{\LCS(S, x)}
    \ge \abs[\Big]{\LCS\bigl( (01)^a, x_1 \bigr)} + \abs[\Big]{\LCS\bigl( \{0^b, 1^b\}, x_2 \bigr)}
    \ge a + b = (n+m) / 3 . \qedhere
  \]
\end{proof}

We observe that taking $m=n$ in \Cref{eq:LCS-cover}, we have a $(2n/3)$-LCS cover consisting of the two strings $(01)^{n/3}0^{n/3}$ and $(01)^{n/3}1^{n/3}$.  This suggests a one-trace algorithm that returns an $n$-bit string $\wh{x}$ that achieves $\smash{\abs{\LCS(\wh{x}, x)} \ge (2/3 - o(1)) n}$ with high probability when $\rho=\omega(\log n/n)$:
to determine which one of the two $n$-bit strings $\smash{(01)^{n/3}0^{n/3}, (01)^{n/3}1^{n/3}}$ to output, it simply needs to determine (with high probability) from the trace $\by\sim\Del_\delta(x)$ whether the majority of the last $2n/3$ bits of the unknown $x$ is $0$ or $1$, which can be done (to accuracy $o(1)$) by simply taking the majority of the last $2\rho n/3$ bits of $y$.
A routine computation now gives the first sentence of \Cref{thm:worst-case-small-rho-informal}.

We further note that the simple $(2n/3)$-LCS cover given by $\{(01)^{n/3}0^{n/3}, (01)^{n/3}1^{n/3} \}$ is essentially  best possible among all covers of constant size; more precisely, 
  for any positive constant $\eps$, any $(2/3+\eps)n$-LCS cover 
  must have size $\Omega(\log n)$.
This is a consequence of a recent result of Guruswami, Haeupler, and Shahrasbi \cite{GHS20}; we give the proof in \Cref{ap:bestLCScover}.

\subsubsection{A $(2/3+\Omega(\rho))$-LCS algorithm for $\rho=\omega(1/n^{1/3})$}

 Next we prove the second part of  \Cref{thm:worst-case-small-rho-informal}.
 It follows from the following theorem:

\begin{theorem} [Worst-case algorithm, small retention rate] 
\label{thm:worst-case-small-rho}
There 
exists an absolute constant $c > 0$ 
such that the following holds. 
Let the retention rate $\rho := \rho(n) = 1-\delta(n)$ such that $\rho=\omega(1/n^{1/3})$.
There is an $O(n)$-time algorithm $A$ which is given as input the values $n, \delta$, and a single trace $\by \sim$ $ \Del_\delta(x)$, where $x \in \zo^n$ is an unknown source string.
With probability at least $1 - e^{-\Omega( \rho^3 n)}$ over the randomness of $\by\sim \Del_\delta(x)$, $A$ outputs a hypothesis string $\wh{\bx} \in \zo^n$ satisfying
\[
  \abs[\big]{\LCS(\wh{\bx},x)}
  \ge (2/3 + c \rho ) \cdot n .
\]
\end{theorem}

An easy computation using the high-probability bound provided by \Cref{thm:worst-case-small-rho} shows that if $\rho \geq \omega(1/n^{1/3})$, then we get that $L_{1,\worst}(\delta,n) \geq (2/3 + \Omega(\rho)) \cdot n$, giving the bound on expected LCS that is claimed in \Cref{thm:worst-case-small-rho-informal}.

The algorithm for \Cref{thm:worst-case-small-rho} improves on the $(2/3-o(1))n$-LCS algorithm
  described in \Cref{sec:LCSalgo}.
The high-level idea is to do better than the $(n+m)/3$ benchmark given by \Cref{claim:lcs-cover} on 
  the $(n/3)$-prefix $x^{(1)}$ of $x$.
For intuition, suppose we could find an $\wh{x}^{(1)} \in \zo^\ast$ such that $$\abs[\big]{\LCS(\wh{x}^{(1)}, x^{(1)})} \ge \frac{\abs{\wh{x}^{(1)}} + \abs{x^{(1)}}}{3} + \eps n.$$
Then we could potentially apply the approach of the one-trace algorithm from the previous subsection on the remaining bits of $x$, and outputs $\wh{x}$ that extends $\wh{x}^{(1)}$ to achieve an LCS of roughly 
\[
  \frac{ \abs{\wh{x}^{(1)}} + \abs{x^{(1)}} }{3} + \eps n
  + \frac{(n - \abs{\wh{x}^{(1)}}) + (n - \abs{x^{(1)}})}{3}
  = \frac{2n}{3} + \eps n .
\]

We now discuss how to beat the $(n+m)/3$ benchmark in more detail.
Let $L=[\rho n/3]$ and $y_L$ be the $(\rho n/3)$-length prefix of the trace $y$.
We divide $y_L$ into blocks of size $2000$.
If a block~contains only 0s, then it is very likely (probability at least, say, 0.9) that there is a corresponding subword in $x$ of size about $2000/\rho$ that contains mostly 0s; such a subword has large LCS (say, at least $1999/\rho$) with the string $0^{2000/\rho}$.
So if most blocks contain only 0s or only 1s (Case 2 in the description of Algorithm $A$ given below), then by outputting an $\wh{x}^{(1)}$ which is a corresponding sequence of $0^{2000/\rho}$'s and $1^{2000/\rho}$'s, such an $\wh{x}^{(1)}$ will have an LCS with $x^{(1)}$ that is much larger than $(\abs{x^{(1)}} + \abs{\wh{x}^{(1)}})/3$.

On the other hand, if most blocks contain both a 0-bit and a 1-bit, then we know that the string $x^{(1)}$ must alternate between 0s and 1s at least $t := \Omega(\rho n)$ times. In this case (Case 1 in the algorithm description), we can use the shorter string $(01)^{n/3 - t}$ to achieve an LCS of size $n/3$ with $x^{(1)}$, which also gives us an $\Omega(\rho n)$ savings.

The rest of \Cref{sec:worst-case-small-rho} gives a formal proof of \Cref{thm:worst-case-small-rho}.

\subsubsection{The Algorithm $A$}
In this subsection we describe the algorithm $A$ to prove \Cref{thm:worst-case-small-rho}.
Let $\gamma:=\rho/720000$. We show that given a trace $\by\sim\Del_\delta(x)$ for any
  unknown $x\in \{0,1\}^n$, the algorithm $A$ returns $\wh{\bx}$ satisfying
\begin{equation}\label{eq:hehe111}
\big|{\LCS(x, \wh{\bx})}\big| \ge \frac{2n}{3} + \frac{\rho n}{90000} - 4 \gamma n
\end{equation}
with probability at least $1-e^{-\Omega(\gamma^2 \rho n)}$.
Setting $c=1/180000$ in \Cref{thm:worst-case-small-rho}
  finishes the proof.

Given a trace $y$ of $x \in \zo^n$,
  $A$ outputs $\wh{x}:=0^n$ if its input trace $y$ has $\abs{y} < (\rho - \gamma) n$.
We refer to this case as Case 0; 
  henceforth we will assume $\abs{y} \ge (\rho-\gamma) n$ below.

Let $L := [\rho n/3]$ 
  and $y_L$ be the first $\rho n /3$ bits of $y$.
Divide $y_L$ into $B := \rho n/6000$ many blocks $y_{L_1}, \ldots, y_{L_B}$ of length $2000$   each (so $L_i := \{2000(i-1) + 1, \ldots, 2000i\}$).
Algorithm $A$ identifies the $y_{L_i}$'s that contain only $0$s or only $1$s.
Specifically, let
\[
  B' := \bigl\{ i \in [B]: \text{$y_{L_i} = z_i^{2000}$ for some $z_i \in \zo$} \bigr\} .
\]
There are two cases:
\begin{flushleft}
\begin{enumerate}
\item[] \noindent{\bf Case 1: $\abs{B'} < 0.8B$.} (In this case, a significant number of blocks are ``not pure.'') 
Let
\begin{align} \label{eq:params-abc}
  c := \frac{n}{3} - \frac{\rho n}{60000},\quad a := \frac{\rho n}{45000}\quad\text{and}\quad b := \frac{n}{3} - \frac{\rho n}{90000} .
\end{align}
Let 
  $z \in \zo$ be the majority of the last $2\rho b$ bits of $y$.
$A$ outputs the $n$-bit  
  $\smash{\wh{x} := (01)^{c + a} z^{b}}$.
\item[] \noindent{\bf Case 2: $\abs{B'} \ge 0.8 B$.}
(Most blocks are ``pure.'')
Let 
  $z\in \{0,1\}$ be the majority of the 
  last $2\rho n/9$ bits of $y$. 
  $A$ outputs the following $n$-bit string
\[
  \wh{x} := \wh{x}^{(1)} \circ (01)^{2n/9}\circ  z^{2n/9},
\]
where $\wh{x}^{(1)}$ is the concatenation of $z_i^{2000/\rho}$ 
  for each $i\in [B]$ with $z_i$ being the bit such that $y_{L_i}=z_i^{2000}$ when $i\in B'$ and $z_i=0$ when $i\notin B'$ 
  (so $\wh{x}^{(1)}$ has length $n/3$).
\end{enumerate}
\end{flushleft}
\subsubsection{Analysis of Algorithm $A$}
Let $x\in \{0,1\}^n$ be the unknown source string. We start by describing an equivalent process of drawing $\by\sim\Del_\delta(x)$.
Let $x_\infty$ be the infinite string obtained from $x$ by padding infinitely many copies of a special symbol $\ast$ at the end.
Consider sampling an infinite subsequence $\by_\infty$ of $x_\infty$ by the following infinite process:
For each round $j=1,2,\dots$, we sample a prefix $\bx^j$ of $x_\infty$ of length $\abs{\bx^j} \sim \Geo(\rho)$, then output the last bit of $\bx^j$ as the $j$-th bit of $\by_\infty$ and delete the prefix $\bx^j$ from $x_\infty$ before moving on to the next value of $j$.
Finally, we set $\by$ to be the 
  longest prefix of $\by_\infty$ that does not contain any special symbol $\ast$.
It is easy to check that $\by$ drawn from this process is identically distributed to $\Del_\delta(x)$.

We introduce some notation for working with $\bx^j$ as a byproduct of the above random process of drawing $\by\sim\Del_\delta(x)$.
For a subset $S \subseteq \N$ (e.g., $L$ introduced in the description of the algorithm), we write $\bx^S$ to denote the concatenation of $\bx^j: j\in S$, where $\bx^j$ is the prefix drawn in the $j$-th round.
Note that the string $\bx^{[\abs{\by}]}$ does not necessarily contain the source string $x$ (it may not contain some of its last few bits) but the string $\bx^{[\abs{\by}+1]}$ always contains $x$ as a prefix.

Let $\by\sim\Del_\delta(x)$ be a trace drawn using the process above, and let 
  $\wh{\bx}$ be the string returned by the algorithm $A$ when running on $\by$.
We say $A$ succeeds (on $\by$) if $\wh{\bx}$ satisfies \Cref{eq:hehe111}
  and $A$ fails otherwise.
It suffices to show that all three probabilities $\Pr_{\by\sim\Del_\delta(x)}[\by\ \text{in Case $0$}]$,
$$
\Pr_{\by\sim \Del_{\delta}(x)}\big[\by\ \text{in Case $1$ and $A$ fails}\big]\quad\text{and}\quad
\Pr_{\by\sim \Del_{\delta}(x)}\big[\by\ \text{in Case $2$ and $A$ fails}\big]
$$
are at most $e^{-\Omega(\gamma^2\rho n)}$. The upper bound for Case 0 follows 
  by the Chernoff bound (which is indeed $e^{-\Omega(\rho n)}$).
Below we analyze the two main cases of the algorithm separately.

\paragraph{Case 1: $\abs{B'} < 0.8 B $.}

Recall from the description of $A$ that we are in Case 1 if $\by\sim \Del_\delta(x)$ has length at least $(\rho-\gamma)n$ and 
  $|\bB'|<0.8 B$. 
Recall  
  from \Cref{eq:params-abc} our choices of $a,b$ and $c$, 
and let $\bz\in\zo$ be the majority of the last $2\rho b$ bits in $\by$.
We partition $x$ into $x^{(1)}\circ x^{(2)}\circ x^{(3)}$ with 
$$
\abs{x^{(1)}} = n/3,\quad \abs{x^{(2)}} = a \quad\text{and}\quad \abs{x^{(3)}} = 2b.
$$

Our goal is to show that 
$$
\Pr_{\by\sim \Del_{\delta}(x)}\bigl[\by\ \text{in Case $1$ and $A$ fails} \bigr]\le e^{-\Omega(\gamma^2\rho n)}.
$$
This follows from the following two claims:
Let $E_1$ denote the event of $|\bx^L|\ge n/3+\gamma n$ and $E_2$ denote the event of 
  $\bz$ appearing less than $b-\gamma n$ many times in $x^{(3)}$.

\begin{claim} \label{claim:events-case1}
For any string $x\in \{0,1\}^n$, we have 
$$
\Pr_{\by\sim \Del_{\delta}(x)}\bigl[\by\ \text{in Case $1$} \land ( E_1\lor E_2) \bigr]\le e^{-\Omega(\gamma^2\rho n)}. 
$$ 
\end{claim}
\begin{claim}\label{claim:events-case12}
The algorithm $A$ succeeds whenever $\by\sim \Del_\delta(x)$ satisfies (1) $\by$ falls in Case 1; (2) $\overline{E_1}$: $|\bx^L|<n/3+\gamma n$; and 
  (3) $\overline{E_2}$: $\bz$ appears at least $b-\gamma n$ many times in $x^{(3)}$.
\end{claim}

\begin{proof}[Proof of \Cref{claim:events-case1}]
It follows from \Cref{claim:neg-binomial} that the probability of $E_1$ alone is at most $e^{-\Omega(\gamma^2\rho n)}$.
So it suffices to upper bound $\Pr_{\by} [ \by\ \text{in Case 1 and $E_2$}]$. 
Assume without loss of generality that $x^{(3)}$ has at least $b+\gamma n$ many $z$'s for some $z\in \{0,1\}$;
otherwise the probability above is trivially $0$.
By Chernoff bound we have 
\begin{align*}
\Pr_{\by\sim \Del_\delta(x)}\big[\text{\# of bits in $x^{(3)}$ that survive in $\by\ge 2\rho b+\rho\gamma n/3$}\big] &\le e^{-\Omega(\gamma^2 \rho n)}\quad\text{and}\\
\Pr_{\by\sim \Del_\delta(x)}\big[\text{\# of $z$'s in $x^{(3)}$ that survive in $\by\le \rho b+2\rho\gamma n/3$}\big] &\le e^{-\Omega(\gamma^2 \rho n)} .
\end{align*}
So with probability at least $1-e^{-\Omega(\gamma^2 \rho n)}$, the number of bits in $x^{(3)}$ that survive
  in $\by$ is at most $2\rho b+\rho\gamma n/3$ and among them at least $\rho b+2\rho\gamma n/3$ bits are $z$.
In this case it cannot happen that $\by$ falls in Case $1$ and $\bz\ne z$.
It follows that $\Pr_{\by} [ \by\ \text{in Case 1 and $E_2$}]\le e^{-\Omega(\gamma^2\rho n)}$.
\end{proof}

\begin{proof}[Proof of \Cref{claim:events-case12}]
When $\by\sim \Del_\delta(x)$ falls in Case 1, the string $\wh{\bx}$ returned by $A$ is 
\begin{align*}
\wh{\bx} := (01)^c\circ (01)^a\circ \bz^{b}.
\end{align*}
We lowerbound $|\LCS(x,\wh{\bx})|$ by
$$
\big|\LCS(x^{(1)},(01)^c)\big|+\big|\LCS(x^{(2)},(01)^a)\big|+\big|\LCS(x^{(3)},\bz^{b})\big|
$$
and below we bound each of the three terms separately. The following simple fact will be useful:


\begin{fact} \label{fact:alternate}
  Suppose $x \in \zo^n$ has $t$ many disjoint $01$'s.
  Then $\LCS(x, (01)^m) = n$ when $m \ge n - t$.
\end{fact}

We start with $|\LCS(x^{(1)},(01)^c)|$.
Because $\abs{\bB'} < 0.8B$, we have that at least $0.2B$ of the $\by_{L_i}$'s contain both $0$ and $1$, and thus there are at least $0.1B$ many disjoint $01$'s appearing in $\by_L$.
Given that $\abs{\bx^L} < n/3 + \gamma n$, the first $n/3 + \gamma n$ bits of $x$ contain at least $0.1B$ many disjoint $01$'s, and so the first $n/3$ bits of $x$ (i.e. $x^{(1)}$) contains at least $0.1B - \gamma n$ many disjoint $01$'s.
Using $c=n/3-0.1B$ and \Cref{fact:alternate}, we have
\[
  \abs[\big]{\LCS(x^{(1)}, (01)^c)}
  = \abs[\bigg]{\LCS\Bigl( x^{(1)}, (01)^{n/3 - 0.1B} \Bigr)}
  \ge \abs[\bigg]{\LCS\Bigl( x^{(1)}, (01)^{n/3 - \left(0.1B - \gamma n \right)} \Bigr) } - 2 \gamma n
  \ge \frac{n}{3} - 2 \gamma n .
\]
Next, given that $x^{(2)}$ only has length $a$ and $x^{(3)}$ 
  contains at least $b-\gamma n$ many $\bz$'s, 
 trivially we have
\[
  \abs[\big]{\LCS(x^{(2)}, (01)^{a})}
  = a\quad\text{and}\quad \abs[\big]{\LCS(x^{(3)}, \bz^b)}\ge b-\gamma n.
\]
It follows that 
$$
\abs[\big]{\LCS(x , \wh{\bx})}\ge \frac{n}{3}+a+b-3\gamma n = \frac{2n}{3}+\frac{\rho n}{90000}-3\gamma n
$$
and $A$ succeeds. This finishes the proof of the claim.
\end{proof}

\paragraph{Case 2: $\abs{B'} \ge 0.8B$.}
Recall that we are in Case 2 if $\by\sim \Del_\delta(x)$ has length at least $(\rho-\gamma)n$ and 
  $|\bB'|\ge 0.8B$.
For each $i\in [B]$ we set $\bz_i$ to be the bit such that $\by_{L_i}=\bz_i^{2000}$ if $i\in \bB'$
  and set $\bz_i=0$ if $i\notin \bB'$.
We also write $\bz$ to denote the majority of the last $2\rho n/9$ bits of $\by$.  
  
The proof proceeds in a similar fashion as Case 1. Let $x=x^{(1)}\circ x^{(2)}\circ x^{(3)}$ with
$$
\abs{x^{(1)}} = n/3,\quad 
\abs{x^{(2)}} = 2n/9 \quad\text{and}\quad
\abs{x^{(3)}} = 4n/9.
$$  
Our goal is to show that 
$$
\Pr_{\by\sim \Del_{\delta}(x)}\big[\by\ \text{in Case $2$ and $A$ fails}\big]\le e^{-\Omega(\gamma^2\rho n)}.
$$
Let $E_1$ denote the event of $|\bx^L|\ge n/3+\gamma n$, $E_2$ denote the event of $\bz$
  appearing less than $2n/9-\gamma n$ many times in $x^{(3)}$,
  and $E_3$ denote the following event:
\begin{enumerate}
\item[] $E_3$: For at least $0.02B$ of $i\in \bB'$, the subword $\bx^{L_i}$
  contains at most $0.9 \cdot 2000/\rho$ many $\bz_i$'s.
\end{enumerate} 
This follows from the following two claims:

\begin{claim} \label{claim:events-case2}
For any string $x\in \{0,1\}^n$, we have 
$$
\Pr_{\by\sim \Del_\delta(x)}\big[\by\ \text{in Case $1 \land ( E_1\lor E_2\lor E_3)$}\big]\le e^{-\Omega(\gamma^2\rho n)}. 
$$
\end{claim}
\begin{claim}\label{claim:events-case22}
The algorithm $A$ succeeds whenever $\by\sim \Del_\delta(x)$ satisfies (1) $\by$ falls in Case 2; (2) $\overline{E_1}$: $|\bx^L|<n/3+\gamma n$; (3) $\overline{E_2}$: $\bz$ appears at least $2n/9-\gamma n$ many times in $x^{(3)}$; and 
  (4) $\overline{E_3}$: At most $0.02B$ of $i\in \bB'$ has $\bx^{L_i}$ contain at most $0.9 \cdot 2000/\rho$ many $\bz_i$'s.
\end{claim}
\begin{proof}[Proof of \Cref{claim:events-case2}]
Events $E_1$ and $E_2$ can be handled similarly as in the proof of \Cref{claim:events-case1}.
Below we show that $\Pr_{\by}[ E_3 ] \le e^{-\Omega(\gamma^2\rho n)}$.
To this end, note that $E_3$ means there are at least $0.02B$ many $i\in [B]$
  such that $\by_{L_i}$ is all $\bz_i$ for some $\bz_i\in \{0,1\}$ while
  $\bx^{L_i}$ has at most $0.9\cdot 2000/\rho$ many $\bz_i$.
  
Let $\bZ_i$ be the indicator random variable for the event above for each $i\in [B]$.
We show below that conditioning on any outcomes of $\bx^1,\ldots,\bx^{2000(i-1)}$,
  the probability of $\bZ_i=1$ is at most $0.01$.
It follows that $E_2$ occurs with probability at most $e^{-\Omega(B)}=e^{-\Omega(\rho n)}$.

For each $i\in [B]$, after fixing any outcomes of $\bx^1,\ldots,\bx^{2000(i-1)}$,
  a necessary condition for $\bZ_i$ to be $1$ is that among the first $0.9\cdot 2000/\rho$ 
  many $0$'s in the current $x_{\infty}$, at least $2000$ of them survive in $\by_{\infty}$, or among the first $0.9\cdot 2000/\rho$
  many $1$'s in $x_{\infty}$, at least $2000$ of them survive in $\by_{\infty}$.
  The probability of $\bZ_i=1$ can be bounded from above by $0.01$ using the Chernoff bound.
  \end{proof}

\begin{proof}[Proof of \Cref{claim:events-case22}]  
When $\by\sim \Del_\delta(x)$ falls in Case 2, the algorithm $A$ returns $$\wh{\bx}=\wh{\bx}^{(1)}\circ (01)^{2n/9}\circ
  \bz^{2n/9} ,$$ where
  $\wh{\bx}^{(1)}$ is the concatenation of $\bz_i^{2000/\rho}$, $i\in [B]$.
We lowerbound $|\LCS(x,\wh{\bx})|$ by
\begin{align*}
&\big|\LCS(x^{(1)},\wh{\bx}^{(1)})\big|+ \big|\LCS(x^{(2)},(01)^{2n/9})\big|+\big|\LCS(x^{(3)},\bz^{2n/9})\big|\\
&\hspace{1.5cm}\ge \big|\LCS(x^{(1)},\wh{\bx}^{(1)})\big| + 2n/9 + 2n/9-\gamma n.
\end{align*}
To bound $|\LCS(x^{(1)},\wh{\bx}^{(1)})|$, we write $\bB''$ to denote the set of $i\in \bB'$
  such that $\bx^{L_i}$ contains at least $0.9 \cdot 2000/\rho$ many $\bz_i$'s.
  It follows from Item (4) in \Cref{claim:events-case22} that $|\bB''|\ge 0.98 \cdot 0.8 B \ge 0.78 B$.
We have 
\begin{align*}
  \abs[\big]{\LCS(x^{(1)}, \wh{\bx}^{(1)})}
  &\ge \abs[\big]{\LCS(\bx^L, \wh{\bx}^{(1)})} - \gamma n  \\[.8ex]
  &\ge \sum_{i \in \bB''} \abs[\big]{\LCS(\bx^{L_i}, \bz_i^{2000/\rho})} - \gamma n \\
  &\ge 0.78\cdot \frac{\rho n}{6000}\cdot 0.9 \cdot \frac{2000}{\rho}-\gamma n\\ & = 0.702\cdot \frac{n}{3}-\gamma n.
\end{align*}
Therefore, altogether we have
\begin{align} \label{eq:lcs-case2}
  \abs[\big]{\LCS(x,\wh{\bx})} 
   \ge \frac{0.702 \cdot 3 + 4}{9} \cdot n - 2 \gamma n \nonumber \ge (0.678 - 2 \gamma) n  
\end{align}
and $A$ succeeds. This finishes the proof of the claim.
\end{proof}


%% file: worst-case-small-rho-upper-bound.tex

\subsection{Bounds on the performance of any one-trace (or few-trace) algorithms}

Complementing \Cref{thm:worst-case-small-rho}, we show that for worst-case approximate trace reconstruction, even if the total number of bits obtained across multiple traces is $n^{0.999}$, it is not possible to achieve expected $\LCS$ of $(2/3 + c)n$ for any constant $c > 0$. The following theorem gives a more detailed version of \Cref{thm:worst-case-small-rho-upper-bound-informal}.


\begin{theorem} [Worst-case upper bound on any few-trace algorithm, small retention rate] \label{thm:worst-case-small-rho-upper-bound}
Let $\kappa>0$ be any absolute constant and let $t(n),\rho(n)=1-\delta(n)$ be such that $t(n) \rho(n) \leq 1/n^{\kappa}$. For sufficiently large $n$, we have 
\[
L_{t(n),\worst}(\delta(n),n) \leq (2/3 + o_n(1))n.\]
\end{theorem}

In order to prove \Cref{thm:worst-case-small-rho-upper-bound}, we first introduce some additional notation.

\subsubsection{Notation} \label{sec:deck-notation}

\noindent {\bf Decks.} 
For $k \in \N$, the \emph{$k$-deck} of a string $z \in \zo^n$, denoted $\deck_k(z)$, is the vector in $\Z^{\zo^k}$ whose $y$-th element (for $y \in \zo^k$) is the number of occurrences of $y$ as a length-$k$ subsequence of $z$. 

Let ${\cal M}$ be a mixture of $n$-bit strings with mixing weights $p_1,\dots,p_m$ on strings $z^1,\dots,z^m \in \zo^n$ (in other words ${\cal M}$ is a distribution over $n$-bit strings). The \emph{$k$-deck of ${\cal M}$}, denoted $\deck_k({\cal M})$, is defined to be the following vector in $\R^{\zo^k}$:
\[
\deck_k({\cal M}) = \sum_{i=1}^m p_i \deck_k(z^i).
\]

Given $y \in \zo^k$ we write $\deck_k(z)_y$ to denote the $y$-th element of $\deck_k(z)$ and $\deck_k({\cal M})_y$ to denote the $y$-th element of $\deck_k({\cal M})$.
Note that for any string $z \in \zo^n$ we have $\sum_{y \in \zo^k} \deck_k(z)_y = {n \choose k}$, and likewise $\sum_{y \in \zo^k} \deck_k({\cal M})_y={n \choose k}$ for any mixture ${\cal M}$ of $n$-bit strings.

\medskip

\noindent {\bf Segments.} 
We view an $n$-bit source string $x \in \zo^n$ as being composed of $n/\ell$ consecutive segments of length $\ell$, for some $\ell = \ell(n)$. 

\medskip

\noindent {\bf Average LCS of a set.} 
Given any set of strings ${\cal S} \subseteq \zo^n$, define
\[
\AvgLCS({\cal S}) := \max_{x' \in \zo^n} {\frac 1 {|{\cal S}|}} \sum_{s \in {\cal S}} |\LCS(x',s)|,
\]
i.e.,~$\AvgLCS({\cal S})$ is the largest possible value (over all possible hypothesis strings $x' \in \zo^n$) of the average LCS between an element of ${\cal S}$ and $x'$. 

We will relate $L_{t(n),\worst}(\delta(n),n)$ to $\AvgLCS({\cal S})$ of a set ${\cal S}$ which is (a slight modification of) the \emph{Bukh--Ma code}, a set of $n$-bit strings that was first studied in \cite{BukhMa14} and further analyzed in \cite{GHS20}.

\subsubsection{The Bukh--Ma code} \label{sec:bukh-ma}

Fix a segment length $\ell=\ell(n)$ which divides $n$. Take $\eps$ to be a suitable $o_n(1)$ value, and let $C_{n,\eps}$ be the Bukh--Ma code analyzed in \cite{GHS20}:  
\begin{equation} \label{eq:Cneps}
C_{n,\eps}=\left\{
(0^r1^r)^{{\frac n {2r}}}: r = {\frac 1 {\eps^{4u}}}, u = 1,\dots,{\frac 1 2} \log_{1/\eps^4} \ell
\right\}.
\end{equation}

We denote the string $(0^r1^r)^{{\frac n {2r}}}$ where $r = {\frac 1 {\eps^{4u}}}$ by $A_u$, for $u = 1,\dots,{\frac 1 2}\log_{1/\eps^4} \ell$.
We remark that for each string $A_u$ in the Bukh--Ma code above, the ``period" $2r = 2/\eps^{4u}$ divides the segment length $\ell$.

\begin{theorem}[Implicit in the proof of \cite{GHS20}, Theorem 1.4] \label{thm:GHS}
    For any $x \in \zo^n$, there can be at most ${\frac {1200}{\eps^3}}$ many strings $A_u \in C_{n,\eps}$ 
    that have $|\LCS(x,A_u)| \geq (2/3 + \eps/6)n.$
\end{theorem}
\newcommand{\adv}{\mathrm{adv}}
\noindent
\emph{Proof sketch:} We explain how \Cref{thm:GHS} is implicit in the proof of Theorem~1.4 of \cite{GHS20}. 
In \cite{GHS20}, it is shown 
(see Section~3, starting after the proof of their Lemma~3.1) that for any $x \in \zo^n$,  if a set of $m$ strings from $C_{n,\eps}$ is such that each of the $m$ strings (call the string $s$) has $\adv(x,s)> \eps/2$, then we must have $m 
\leq 1200/\eps^3$. 
Since
$\adv(x,s) = {\frac {3 |\LCS(x,s)| - |x|-|s|}{|x|}}$
(see \cite{GHS20}'s Definitions~2.4 and 2.5),
having $\adv(x,s) > \eps/2$ is 
equivalent to having $|\LCS(x,s)| \geq (2/3 + \eps/6)n$. 
\qed

\medskip

Fix $x \in \zo^n$. Using \Cref{thm:GHS}, we can upper bound the average $\LCS$ of $x$ with $C_{n,\eps}$ by 
\begin{equation} \label{eq:twothirds}
    {\frac 1 {|C_{n,\eps}|}} \sum_{s \in C_{n,\eps}} |\LCS(x,s)| \leq {\frac {2 \cdot 1200/\eps^3}{\log_{1/\eps^4} \ell}} \cdot n + (2/3 + \eps/6)n = (2/3 + o(1))n.
\end{equation}
As this is true for all $x \in \zo^n$, we conclude that 
\begin{equation}
\label{eq:avglcs}
\AvgLCS(C_{n,\eps}) \leq (2/3 + o(1))n.
\end{equation}

\subsubsection{Relating $L_{t(n),\worst}(\delta(n),n)$ to $\AvgLCS({\cal S})$}

The following claim will allow us to upper bound the performance of any algorithm that receives $t = t(n)$ traces at deletion rate $\delta(n)$ by (essentially) $\AvgLCS({\cal S})$ for any set ${\cal S}$ satisfying certain properties.


\begin{claim} \label{claim:mixture}
Let $\ell$ be such that both $\ell$ and $n^\ell$ are at least $n^c$ for some positive constant $c$.
Let $\calS_\ell = \{s^{(1)}_\ell, s^{(2)}_\ell, \dots, s^{(m)}_\ell\} \subset \zo^\ell$ be a set of $\ell$-bit strings. 
Define the set of $n$-bit strings $\calS_n = \{s^{(1)}_n, s^{(2)}_n, \cdots, s^{(m)}_n\} \subset \zo^n$, where each 
    string $s^{(u)}_n$ is constructed by concatenating $n/\ell$ copies of $s^{(u)}_\ell$. 
    For each $u \in [m]$ let ${\cal M}^{(u)}$ be a mixture of $\ell$-bit strings with the following properties:
    
    \begin{enumerate}
    
    \item With probability $1-o(1)$, a random $\ell$-bit string $\bz$ drawn from ${\cal M}^{(u)}$ has $\LCS(\bz,s^{(u)}_\ell) \geq (1-o(1))\ell$;
    
    \item For each $u \in [m]$ the $k$-deck $\deck_k({\cal M}^{(u)})$ is the same.
    
    \end{enumerate}
    
    Let $\rho(n) = 1 - \delta(n)$. Then we have
    \[
        L_{t(n),\worst}(\delta(n),n) \leq t(n) \cdot \ell^k \cdot \rho(n)^{k+1} \cdot n^2 + \AvgLCS(\calS_n) + o(n).
    \]
\end{claim}

\begin{proof}
    Let ${\cal M}$ be the following distribution over $n$-bit strings: to draw $\bx \sim {\cal M}$, first draw a uniform $\bu \sim [m]$, then independently draw $n/\ell$ many $\ell$-bit strings $\bx^{(1)},\dots,\bx^{(n/\ell)} \sim {\cal M}^{(\bu)}$, and concatenate them to yield $\bx = \bx^{(1)} \cdots \bx^{(n/\ell)}.$

    Let $A$ be any algorithm that takes as input $t := t(n)$ traces $\by^{(1)}, \cdots, \by^{(t)}$ of $\bx$, and outputs an $n$-bit hypothesis string. We suppose that in addition to the input traces, $A$ is also told, for each trace, how many bits of the trace come from each of the $n/\ell$ segments of the source string; we upper bound $L_{t(n),\worst}(\delta(n),n)$ by upper bounding the performance of any algorithm that also receives this extra auxiliary information.
    
    The probability that any of the $n/\ell$ many $\ell$-bit segments of $\bx$ has at least $k + 1$
    bits from it surviving into any of the $t$ traces is at most $t \cdot (n/\ell) \cdot (\rho(n) \ell)^{k+1} = t n \ell^k \rho(n)^{k+1}$. In this case we trivially upper bound the $\LCS$ between $\bx$ and the output of $A$ by $n$.

    Otherwise, at most $k$ bits survive from each segment in each trace. The distribution of these bits is the same, regardless of the random $\bu \sim [m]$ chosen in the construction of $\bx$. This follows from property (2.) above and the easily observable fact that if the $k$-deck $\deck_k({\cal M}^{(u)})$ is the same for each $u \sim [m]$, then the $k'$-deck $\deck_{k'}({\cal M}^{(u)})$ is also the same for each $u \sim [m]$, for all $k' \leq k$. In this case, the optimal string for algorithm $A$ to output is the $n$-bit string $x^*$ that achieves $\AvgLCS(\calS_n)$.

    By property (1.) above and a standard Chernoff bound, with $1-o(1)$ probability we have that a $1-o(1)$ fraction of the $n/\ell$ strings $\bx^{(1)}, \cdots, \bx^{(n/\ell)}$ drawn from $\calM^{(\bu)}$ satisfy $|\LCS(\bx^{(i)}, s^{(\bu)}_\ell| \geq (1-o(1))\ell$, 
    so with $1-o(1)$ probability the string $\bx=\bx^{(1)}\cdots \bx^{(n/\ell)}$ has $|\LCS(\bx,s^{(\bu)}_n)| \geq (1-o(1))n$.
    

    Recall that $x^* \in \zo^n$ is the string achieving $\AvgLCS(\calS_n)$. We will use the triangle inequality on the edit distance $d_{edit}(z,z') := n - |\LCS(z,z')|$ (which is a metric). We have
    \[
      d_{edit}(\bx, x^*) \geq d_{edit}(x^*, s^{(\bu)}_n) - d_{edit}(s^{(\bu)}_n, \bx).
    \]
    Rewriting this inequality in terms of $\LCS$, we have
    \[
       | \LCS(\bx, x^*)| \leq |\LCS(x^*, s^{(\bu)}_n)| + n - |\LCS(\bx,s^{(\bu)}_n)| \leq |\LCS(x^*, s^{(\bu)}_n)| + o(n).
    \]
    We emphasize that $\bx$ is a function of the random $\bu \sim [m]$, while $x^*$ is independent of $\bu$. Taking expectation over $\bu$, we get that
    \[
        \E_{\bu}[|\LCS(\bx, x^*)|] \leq \AvgLCS(\calS_n) + o(n).
    \]
    Combining the two cases above, we obtain the lemma.
\end{proof}

\begin{proof}[Proof of \Cref{thm:worst-case-small-rho-upper-bound} using \Cref{claim:mixture}] In \Cref{lem:good-mixture} below, for any constant $k \in \N$, we will exhibit a set of mixtures ${\cal M}^{(u)}$ satisfying the properties in \Cref{claim:mixture}, with the set $\calS_n$ being $C_{n,\eps}$. Choosing $k = 4/\kappa$ (constant), $\ell = n^{1/k}$, and using the fact that $\AvgLCS(C_{n,\eps}) \leq (2/3 + o(1))n$ (recall \Cref{eq:avglcs}), we conclude that
    \begin{align*}
        L_{t(n),\worst}(\delta(n),n)
        &\leq t(n) \, \ell^k \, \rho(n)^{k+1} \, n^2 + \AvgLCS(\calS_n) + o(n)\\
        &\leq (t(n) \rho(n))^{k+1} n^3 + (2/3 + o(1))n\\
        &\leq n^{3 - (k+1)\kappa} + (2/3 + o(1))n\\
        &\leq (2/3 + o(1))n. \qedhere
    \end{align*}
\end{proof}

%% file: general-k-mixture2.tex

\subsection{Construction of $\calM$ satisfying \Cref{claim:mixture} for any constant $k$}

Let ${\cal S}_\ell$ be the set of $m:={\frac 1 2} \log_{1/\eps^4} \ell$ many $\ell$-bit strings

\[
{\cal S}_\ell = \left\{
(0^{1/\eps^{4u}}1^{1/\eps^{4u}})^{\ell/(2/\eps^{4u})}
\right\},
\quad u=1,\dots,m.
\]

Fix any positive integer $k$ (which should be thought of as a fixed constant, while $\ell \to \infty$).
In this section we construct a collection of $m$ mixtures ${\cal M}^{(1)},\dots,{\cal M}^{(m)}$, where each ${\cal M}^{(u)}$ is a mixture of $\ell$-bit strings, which meet the conditions required by \Cref{claim:mixture}. 
In more detail, we show that the mixtures ${\cal M}^{(1)},\dots,{\cal M}^{(m)}$ that we construct satisfy the following:

\begin{lemma} \label{lem:good-mixture}
For each $u \in [m]$ we have the following:

\begin{enumerate}

\item  With probability $1-o_\ell(1)$, a random $\ell$-bit string $\bz$ drawn from ${\cal M}^{(u)}$ has 
\[
  \abs[\bigg]{\LCS \Bigl(\bz,(0^{1/\eps^{4u}}1^{1/\eps^{4u}})^{\ell/(2/\eps^{4u})} \Bigr)} \geq (1-o_\ell(1))\ell.
\]

\item For each $u \in [m]$ the $k$-deck $\deck_k({\cal M}^{(u)})$ is the same.

\end{enumerate}

\end{lemma}

\noindent {\bf The mixture ${\cal M}^{(u)}$.}
Fix $u \in [m]$ and let $r_0 := 1/\eps^{4u}$.
For $t$ dividing $\ell$, let $x^{(t)}$ denote the $\ell$-bit string
\[
x^{(t)} := (0^t 1^t)^{\ell/(2t)},
\]
so $x^{(r_0)}$ is the $u$-th string $(0^{1/\eps^{4u}}1^{1/\eps^{4u}})^{\ell/(2/\eps^{4u})}$ in ${\cal S}_\ell$.
The mixture ${\cal M}^{(u)}$ will be supported on $k$ strings in $\zo^\ell$,
\[
\supp({\cal M}^{(u)}) = \{x^{(r_0)}, x^{(r_1)}, \dots, x^{(r_{k-1})}\},
\]
where $r_1,\dots,r_{k-1}$ are values that will satisfy $r_0 \ll r_1 \ll \cdots \ll r_{k-1} \ll \ell$  and that will be specified later.
The mixing weight $p_j$ on the $j$-th string $x^{(r_j)}$ will be chosen so that  each $p_j \geq 0$,
 $\sum_{j=0}^{k-1} p_j = 1$ (so ${\cal M}^{(u)}$ is indeed a valid distribution), and 
$p_0 = 1-o_\ell(1)$, which gives item (1) of \Cref{lem:good-mixture}.

To achieve item (2) of \Cref{lem:good-mixture} we will carefully choose the weights $p_0,\dots,p_{k-1}$ so that for each $y \in \zo^k$, the value
$\deck_k({\cal M}^{(u)})_y$ is a function only of $\ell$ (and in particular is independent of the value of $u$).  Towards this end, let us begin to analyze the $k$-deck of a single string $x^{(t)}$. The following is easily verified:

\begin{claim} \label{claim:k-deck-xt}
Fix any $y \in \zo^k$. The value $\deck_k(x^{(t)})_y$ is of the form
\begin{equation} \label{eq:dkxty}
\deck_k(x^{(t)})_y = \sum_{i=0}^{k-1} t^i f_{y,i}(\ell) 
\end{equation}
for some polynomials $f_{y,0}(\ell),\dots,f_{y,k-1}(\ell)$.
\end{claim}

From \Cref{eq:dkxty} we immediately get that
\begin{equation}
\label{eq:dkMjy}
\deck_k({\cal M}^{(i)})_y = \sum_{j=0}^{k-1} p_j \left(\sum_{i=0}^{k-1} r_j^i f_{y,i}(\ell) \right)
=
\sum_{i=0}^{k-1} \left( \sum_{j=0}^{k-1} p_j r_j^i \right) f_{y,i}(\ell).
\end{equation}
Recall that $r_0=1/\eps^{4u}$, so clearly $r_0$ depends on $u$, and that we have yet to choose $r_1,\dots,r_{k-1}$. \Cref{eq:dkMjy}  leads us to consider the following linear system:
\begin{equation}
\label{eq:lin-system}
V p = b
\end{equation}
where $V$ is the $k \times k$ Vandermonde matrix whose rows and columns we index by $i \in \{0,\dots,k-1\}$ and $j \in \{0,\dots,k-1\}$,
\begin{equation} \label{eq:vandermonde}
V_{i,j} = r_j^i,
\end{equation}
and $p$ and $b$ are $k \times 1$ 
column vectors
\newcommand{\myvec}[1]{\ensuremath{\begin{pmatrix}#1\end{pmatrix}}}
\[
p = 
\myvec{p_0 \\ \vdots \\p_{k-1}},
\quad \quad \quad
b = 
\myvec{b_0\\ \vdots \\ b_{k-1}}.
\]
We will prove the following claim:

\begin{claim}
\label{claim:indep-of-u}
There are values $b_0,\dots,b_{k-1}$ 
that have no dependence on $u$ so that the solution
\begin{equation}
\label{eq:solution}
p = V^{-1} b
\end{equation}
to the system (\ref{eq:lin-system}) has each $p_j \geq 0$, $\sum_{j=0}^{k-1} p_j = 1$, and $p_0 = 1 - o_\ell(1).$ 
\end{claim}

By \Cref{eq:dkMjy} this means that the $k$-deck 
\[
\deck_k({\cal M}^{(i)})_y  = \sum_{i=0}^{k-1} b_i f_{y,i}(\ell),\quad \quad
y \in \zo^k,
\]
has no dependence on $u$, giving item (2) of \Cref{lem:good-mixture} and completing its proof.
It thus remains to prove \Cref{claim:indep-of-u}.

\subsubsection{Proof of \Cref{claim:indep-of-u}}

We start by recalling an explicit formula for the inverse of a Vandermonde matrix:

\begin{fact} [\cite{Turner66}]
Let $V=(V_{ij})_{i,j \in \{0,\dots,k-1\}}$ be the $k \times k$ Vandermonde matrix $V_{i,j} = r_j^i$ as specified in \Cref{eq:vandermonde}.
Let $e^{(i)}_j$ be the $j$-th elementary symmetric polynomial on the $k-1$ variables $r_0,\dots,r_{i-1},r_{i+1},\dots,r_{k-1}$.  
Then the inverse matrix $V^{-1}$ is given by
\begin{equation} 
\label{eq:inverse}
V^{-1}_{i,j} = {\frac {(-1)^{j} \cdot e^{(i)}_{k-1-j}}
{\prod_{s \neq i} (r_s - r_i)}}.
\end{equation}
\end{fact}
It will be convenient for us to rewrite \Cref{eq:inverse} in a way which makes the denominator always positive (recall that we will have $r_0 \ll r_1 \ll \cdots \ll r_{k-1}$. Doing this, we obtain
\begin{equation}
\label{eq:inverse2}
V^{-1}_{i,j} = {\frac {(-1)^{i+j} \cdot e^{(i)}_{k-1-j}}
{
\left(
\prod_{0 \leq s \leq i-1} (r_i - r_s)
\right)
\cdot
\left(
\prod_{i+1 \leq s \leq k-1} (r_s - r_i)
\right)
}
},
\end{equation}
and consequently we have that
\begin{equation}
\label{eq:p}
p = V^{-1} b, \quad \text{where for $i = 0,\dots,k-1$,} \quad
p_i = 
{\frac
{
\sum_{j=0}^{k-1} (-1)^{j+i} e^{(i)}_{k-1-j} \cdot b_j
}
{
\left(
\prod_{0 \leq s \leq i-1} (r_i - r_s)
\right)
\cdot
\left(
\prod_{i+1 \leq s \leq k-1} (r_s - r_i)
\right)
}
}
\end{equation}
(note that the denominator of \Cref{eq:p} is independent of $j$).

We now choose $r_j, b_j: j \in [k-1]$ appropriately and show that the $p_i$'s satisfy the conditions in \Cref{claim:indep-of-u}.
Recall that $m = \frac{1}{2} \log_{1/\eps^4} \ell$, so $r_0 = 1/\eps^{4u} \leq1/\eps^{4m} = \sqrt{\ell}$.
For $j \in [k-1]$, we define
\[
  b_j := \frac{1}{(\log\log\ell)^j} \cdot \prod_{s=1}^j r_j \quad\text{and}\quad 
  r_j := \ell^{2/3} \cdot (\log \ell)^j 
\]
(observe that $r_0$ is already fixed to $1/\eps^{4u}$, and that the first row of the Vandermonde matrix system of equations is all-1's, which means that  $b_0 = p_0 + \cdots + p_{k-1}=1$).
These settings are chosen so that in the summation in the numerator of the expression for $p_i$ in \Cref{eq:p}, the $(j=i)$-th term, which is always positive, dominates the sum of the rest of the terms in magnitude.
Specifically, we will show that for $j < i$, the quantity $e^{(i)}_{k-1-j} b_j$ is at most $O((\log\ell)^{-1}) \cdot \prod_{s=1}^{k-1} r_s$ and for $j \ge i$, we have $e^{(i)}_{k-1-j} b_j = (\log\log\ell)^{-j} (1 + o_\ell(1)) \prod_{s=1}^{k-1} r_s$.
So the numerator is at least $(\log\log \ell)^{-i} (1 - o_\ell(1)) \prod_{s=1}^{k-1} r_s \ge 0$, and thus the $p_i$'s are positive because the denominator is positive.
Moreover, the denominator of $p_0$ in \Cref{eq:p} is at most $\prod_{s=1}^{k-1} r_s$.
This shows $p_0 = 1 - o_\ell(1)$.

We now give the full calculation.
First observe that for every $0 \le j \le k-2$ and every $S \subseteq [k-1]$ of size $k-1-j$ not equal to $\{j+1, \ldots, k-1\}$, we have
\begin{equation}
\prod_{s \in S} r_s
  \leq \ell^{2\abs{S}/3} \cdot (\log \ell)^{\sum_{s \in S} s}
  \le \ell^{2\abs{S}/3} \cdot (\log \ell)^{(\sum_{s=j+1}^{k-1} j) - 1}
  = \frac{1}{\log \ell} \prod_{s=j+1}^{k-1} r_s .
  \label{eq:lowerorder}
\end{equation}
So for $j < i$ we have $i \in \{j+1, \ldots, k-1\}$ and so the (positive) quantity $  e_{k-1-j}^{(i)} \cdot b_j$ is ``small,'' i.e. at most $O((\log\ell)^{-1}) \cdot \prod_{s=1}^{k-1} r_s$:
\begin{align}
  e_{k-1-j}^{(i)} \cdot b_j
  &= \biggl( \sum_{\substack{S \subseteq \{0,\dots,k-1\} \setminus \{i\} \\ \abs{S} = k-1-j}} \prod_{s \in S} r_s \biggr) \cdot \biggl( (\log\log\ell)^{-j} \prod_{s=1}^j r_j \biggr) \nonumber \\
  &\le \biggl( \binom{k-1}{j} \frac{1}{\log \ell} \prod_{s=j+1}^{k-1} r_s \biggr) \cdot \biggr( (\log\log\ell)^{-j} \prod_{s=1}^j r_s \biggr) \nonumber \\
  &\le (\log\log\ell)^{-j} \cdot \biggl( \prod_{s=1}^{k-1} r_s \biggr) \cdot \frac{k^j}{\log \ell}
\end{align}
For $j = i$ the (positive) quantity $  e_{k-1-j}^{(i)} \cdot b_j$ is ``large,'' i.e. at least $(\log\log\ell)^{-i} \prod_{s=1}^{k-1} r_s$;
more precisely,
we have
\begin{align}
  e_{k-1-j}^{(i)} \cdot b_j
  &= \biggl( \sum_{\substack{S \subseteq \{0,\dots,k-1\} \setminus \{i\} \\ \abs{S} = k-1-j}} \prod_{s \in S} r_s \biggr) \cdot \biggl( (\log\log\ell)^{-j} \prod_{s=1}^j r_j \biggr) \nonumber \\
  &\ge \biggl(  \prod_{s=j+1}^{k-1} r_s \biggr) \cdot \biggl( (\log\log\ell)^{-j} \prod_{s=1}^j r_j \biggr) \nonumber \\
  &=  (\log\log\ell)^{-j} \cdot \biggl( \prod_{s=1}^{k-1} r_s \biggr).
\end{align}
For $j > i$  the (positive) quantity $  e_{k-1-j}^{(i)} \cdot b_j$ is again ``small,'' i.e. $(\log\log\ell)^{-j} (1 + o_\ell(1)) \prod_{s=1}^{k-1} r_s$:
\begin{align}
  e_{k-1-j}^{(i)} \cdot b_j
  &= \biggl( \prod_{s=j+1}^{k-1} r_s + \sum_{\substack{S \subseteq \{0,\dots,k-1\} \setminus \{i\} \\ \abs{S} = k-1-j \\ S \ne \{j+1, \ldots, k-1\}}} \prod_{s \in S} r_s \biggr) \cdot \biggl( (\log\log\ell)^{-j} \prod_{s=1}^j r_j \biggr) \nonumber \\
  &\le  \Biggl( \biggl( \prod_{s=j+1}^{k-1} r_s \biggr) \cdot \biggl( 1 + \binom{k-1}{j} \frac{1}{\log\ell} \biggr) \Biggr) \cdot \biggl( (\log\log\ell)^{-j} \prod_{s=1}^j r_j \biggr) \nonumber \\
  &= (\log\log\ell)^{-j} \cdot \biggl( \prod_{s=1}^{k-1} r_s \biggr) \cdot \biggl(1 + \frac{k^j}{\log \ell} \biggr) .
\end{align}
Therefore for every $i \in \{0, \ldots, k-1\}$, the alternating sum is dominated by the contribution from $j=i$: more precisely, we have
\begin{align*}
  \sum_{j=0}^{k-1} (-1)^{j+i} e_{k-1-j}^{(i)} \cdot b_j
  &\ge \biggl( \prod_{s=1}^{k-1} r_s \biggr) \Biggl(\sum_{j=i}^{k-1} (-1)^{i+j} (\log\log\ell)^{-j} - \frac{1}{\log\ell} \sum_{\substack{0 \le j \le k-1 \\ j \ne i}} \left(\frac{k}{\log\log\ell}\right)^j \Biggr) \\
  &\ge \biggl( \prod_{s=1}^{k-1} r_s \biggr) \biggl( (\log\log\ell)^{-i} \left(1 - \frac{1}{\log\log\ell} \right) - \frac{2}{\log \ell} \biggr) \\
  &\ge 0 .
\end{align*}
Since, as noted earlier, the denominator of \Cref{eq:p} is positive, this shows that $p_i \ge 0$ for every  $i \in \{0, \ldots, k-1\}$.
Moreover, we have
\begin{align*}
  p_0
  &= \frac{ \sum_{j=0}^{k-1} (-1)^{j} e_{k-1-j}^{(i)} \cdot b_j }{ \prod_{s=1}^{k-1} (r_s - r_0) } \\
  &\ge \frac{ \biggl( \prod_{s=1}^{k-1} r_s \biggr) \Bigl( 1 - \frac{1}{\log\log\ell} - \frac{2}{\log\ell} \Bigr) } { \prod_{s=1}^{k-1} r_s } \\
  &\ge 1 - \frac{2}{\log\log\ell}.
\end{align*}
This completes the proof of \Cref{claim:indep-of-u}. \qed

%% file: worst-case-small-delta3.tex

\section{Worst-case one-trace reconstruction, medium and small deletion rate} \label{sec:worst-case-small-delta}

In this section we consider the medium and small deletion rate regime. In particular, throughout this section we suppose that $\delta \le 1/2$. (Note that if $\delta > 1/2$, then the quantity $1-\delta + \delta^2/2 -\delta^3/2 + \delta^4/2 - \delta^5/2$ is less than $2/3$, so the performance guarantee given by \Cref{thm:worst-case-small-delta} is weaker than the guarantee given by \Cref{thm:worst-case-small-rho-informal} / \Cref{thm:worst-case-small-rho}.)

\subsection{An efficient algorithm achieving expected LCS $(1-\delta + \delta^2/2 -\delta^3/2 + \delta^4/2 - \delta^5/2 - o(1))n$}

As mentioned in the introduction, it is very easy for a one-trace algorithm to achieve expected LCS at least $(1-\delta)n$: this can be accomplished simply by having the hypothesis string $\wh{x}$ be any string that contains the input trace $\by$ as a subsequence.  The expected LCS of such a hypothesis string is clearly at least $\E[|\by|]$, which is $(1-\delta)n$ by linearity of expectation.

The following theorem shows how to improve on this naive bound:

\begin{theorem} [Worst-case algorithm, small deletion rate]\label{thm:worst-case-small-delta}
Let $\delta=\delta(n) \leq 1/2$ be the deletion rate. There is an $O(n)$-time (randomized) algorithm \textsc{Small-rate-reconstruct} which is given as input the values $n$ and $\delta$ and a single trace $\by \sim \Del_\delta(x)$, where $x \in \zo^n$ is an unknown and arbitrary source string.
For any $\gamma \leq 1$, \textsc{Small-rate-reconstruct} outputs a hypothesis string $\wh{\bx} \in \zo^n$ which satisfies
\[
\E\big[|\LCS(\wh{\bx},x)|\big] \geq 
    \Bigl( 1 - e^{-\Omega(\gamma^2 n)} \Bigr) \left( 1 - \delta + \frac{\delta^2}{2} - \frac{\delta^3}{2} + \frac{\delta^4}{2} - \frac{\delta^5}{2} \right)n - 3 \gamma n
\]
(so in particular, taking $\omega(1/\sqrt{n})\leq \gamma \leq o(1)$, we get that the expected value of $|\LCS(\wh{\bx},x)|$ is at least $(1 - \delta + \delta^2/2 - \delta^3/2 + \delta^4/2 - \delta^5/2 - o(1))n$).
\end{theorem}

\begin{algorithm}[t]
\caption{\textsc{Small-rate-reconstruct}} 
\begin{algorithmic}[1] \label{alg:small-rate-reconstruct}
\State Set $j=1$ and $p_y = 1.$
\State  While $p_y \leq |\by|$ do:\\
~~~~~With probability $1-\delta$ set $\widehat{\bx}_j := \by_{p_y}$ and increment $p_y$; \\ ~~~~~with the remaining $\delta$ probability set $\widehat{\bx}_j$ to a uniform bit from $\zo$.\\
~~~~~Set $j := j+1.$
\State If $|\wh{\bx}|<n$ then append $0^{n-|\wh{\bx}|}$ to $\wh{\bx}$, and if $|\wh{\bx}|>n$ then delete bits $\wh{\bx}_{n+1},\dots$ from $\wh{\bx}.$
\State Output the $n$-bit string $\wh{\bx}.$
\end{algorithmic}
\end{algorithm} 

\noindent {\bf Intuition.}
The algorithm \textsc{Small-rate-reconstruct} is given as \Cref{alg:small-rate-reconstruct}.
To analyze the algorithm it is convenient to consider the string $\wh{\bx}'$ which is $\wh{\bx}$ immediately before Step~5 is performed (i.e. with no padding with 0s or deletion applied).  We will show later that $\wh{\bx}'$ is with high probability ``very close to $\wh{\bx}$'', so we can chiefly reason about $\wh{\bx}'$ and take care of the minor difference between $\wh{\bx}'$ and $\wh{\bx}$ at the end of the argument.

We first observe that $\wh{\bx}'$ clearly contains $\by$ as a subsequence.
The main idea of the proof is that a non-negligible fraction of the times that Step~3 is performed, one or more uniform random bits from $\zo$ will be placed in between consecutive bits $\by_{p_y}$ and $\by_{p_y+1}$ in creating the hypothesis string $\wh{\bx}'$ exactly when one or more bits of $x$ were deleted in between $\by_{p_y}$ and $\by_{p_y+1}$ in the creation of the trace $\by$. Each time this happens there is a 1/2 chance that at least one ``additional bit'' beyond the subsequence $\by$ of $\wh{\bx}'$ can be incorporated into a matching between $x$ and $\wh{\bx}'$.
This is the source of the extra $(\delta^2/2 - \delta^3/2 + \delta^4/2 - \delta^5/2)n$ in the LCS bound. {Intuitively, the number of additional bits between every $\by_{p_y}$ and $\by_{p_y + 1}$ in each of $x$ and $\wh{\bx}'$ is distributed according to $\Geo(1-\delta) - 1$, so there are at least $\min\{\Geo(1-\delta)-1, \Geo(1-\delta)-1\} = \Geo(1-\delta^2)-1$ many additional bits between $\by_{p_y}$ and $\by_{p_y+1}$ in \emph{both} of $x$ and $\wh{\bx}'$, and there is a 1/2 chance each uniform additional $\zo$ bit in $\wh{\bx}'$ matches an additional bit in $x$}.

We now provide formal details. 
\begin{proof}[Proof of \Cref{thm:worst-case-small-delta}]
Let $x\in \{0,1\}^n$ be the unknown source string.
  Consider appending infinitely many copies of a special symbol $\ast$ to $x$ to form an infinite string $x_\infty$.
  We sample an infinite subsequence $\by_\infty$ of $x_\infty$ by the following infinite process:
  For each $j=1,2,\dots$, we sample a prefix $\bx^j$ of $x_\infty$ of length $\abs{\bx^j} \sim \Geo(1-\delta)$, then output the last bit of $\bx^j$ as the $j$-th bit of $\by_\infty$ and delete the prefix $\bx^j$ from $x_\infty$ before moving on to the next value of $j$.

  Note that the longest prefix of $\by_\infty$ that does not contain any $x_i : i > n$ is identically distributed as the trace $\by \sim \Del_\delta(x)$.
  Equivalently, the concatenation of the last bit in each of $\bx^1, \ldots, \bx^{\abs{\by}}$ is identically distributed as $\by \sim \Del_\delta(x)$.

  Let $\bT=\{t_1 < \cdots < t_{|\by|}\}$ be the set of $|\by|$ many locations $j \in \{1,2,\dots,\}$ such that $\wh{\bx}'_j$ was set to $\by_{p_y}$ in some execution of Step~3 of \textsc{Small-rate-reconstruct}.
  Note that the elements of $\bT$ are the indices of the $\bT$race bits in $\wh{\bx}'$ and that $t_i$ is the index such that $\wh{\bx}'_{t_i}$ was set to $\by_{i}$ in Step~3 of the execution of \textsc{Small-rate-reconstruct}.
  Let $$\hat{\bx}'^1 := \hat{\bx}'_{[1:t_1]}, \hat{\bx}'^2 := \hat{\bx}'_{[t_1+1:t_2]}, \ldots, \wh{\bx}'^{\abs{\by}} := \wh{\bx}'_{[t_{\abs{\by}-1}+1:t_{\abs{\by}}]}.$$
  Observe that for each $i \in [\abs{\by}]$, both $\bx^i$ and $\wh{\bx}'^i$ are identically distributed.
  In particular, their lengths $\abs{\bx^i}$ and $\abs{\wh{\bx}'^i}$ are distributed according to $\Geo(1-\delta)$, and so the minimum of both lengths, i.e.  $\bd_i := \min\{\abs{\bx^i}, \abs{\wh{\bx}'^i}\}$, is distributed according to $\Geo(1-\delta^2)$.
  Moreover, the last bits in $\bx^i$ and $\wh{\bx}'^i$ are equal to $\by_i$, and the rest of the bits in $\wh{\bx}^i$ are uniform.
  
  For each $i \in [\abs{\by}]$, since the length-$(\bd_i-1)$ prefix of $\wh{\bx}'^i$ is random, in expectation (over the randomness of $\wh{\bx}'^i$) it agrees with the length-$(\bd_i-1)$ prefix of $\bx^i$ on $(\bd_i-1)/2$ of the bits.
  Also, the last bit of both $\bx^i$ and $\wh{\bx}'^i$ are the same.
  Hence, we have
  \[
    \E_{\wh{\bx}'^i} \Bigl[ \abs{\LCS(\bx^i, \wh{\bx}^i)} \Bigr] \ge (\bd_i-1)/2 + 1.
  \]
  Observe that the concatenation of $\bx^i : i\in [\abs{\by}]$ is a prefix of $x$, because the last bit of $\bx^{\abs{\by}}$ is the last bit of $\by \sim \Del_\delta(x)$, and the concatenation of $\wh{\bx}'^i : i \in [\abs{\by}]$ is exactly $\wh{\bx}'$.
  Thus,
  \[
    \E_{\wh{\bx}'} \Bigl[ \abs{\LCS(x, \wh{\bx}')} \Bigr]
    \ge \sum_{i=1}^{\abs{\by}} \Bigl[ \E_{\bx^i, \wh{\bx}'^i} \abs{\LCS(\bx^i,\wh{\bx}'^i)} \Bigr] \\
    \ge \sum_{i=1}^{\abs{\by}} \Bigl( (\bd_i-1)/2 + 1 \Bigr)
    = \abs{\by} + \sum_{i=1}^{\abs{\by}} (\bd_i-1)/2 .
  \]
  Since $\bd_i \sim \Geo(1-\delta^2)$, we have $\E[\bd_i-1] = \frac{1}{1 - \delta^2} - 1 = \frac{\delta^2}{1 - \delta^2}$.
  So taking expectation over $\abs{\by}$, we obtain 
  \begin{align*}
    \E \Bigl[ \abs{\LCS(x, \wh{\bx}')} \Bigr]
    &\ge \E\bigl[ \abs{\by} \bigr] + \E \bigl[ \abs{\by} \bigr] \frac{\delta^2}{2(1-\delta^2)} \\
    &= (1-\delta) n \cdot \left(1 + \frac{\delta^2}{2 (1-\delta^2)} \right) \\
    &= \left( 1 - \delta + \frac{\delta^2}{2(1 - \delta^2)} - \frac{\delta^3}{2(1-\delta^2)} \right) n \\
    &\ge \left( 1 - \delta + \frac{\delta^2}{2} - \frac{\delta^3}{2} + \frac{\delta^4}{2} - \frac{\delta^5}{2} \right) n .
  \end{align*}
  To finish the proof, we relate $\E[\abs{\LCS(\wh{\bx}, x)}]$ to $\E[\abs{\LCS(\wh{\bx}', x)}]$.
  We observe that
  \[
  |\LCS(\wh{\bx},x)| \geq \max\{0, |\LCS(\wh{\bx}',x)| - \bk\},
  \]
  where $\bk$ is the number of bits $\wh{\bx}_{n+1},\dots$ deleted from $\wh{\bx}$ in Step~5 of \textsc{Small-rate-reconstruct} if bits are deleted in that step (and $\bk=0$ otherwise).  So it remains only to show that with high probability $\bk$ is small.  
  
  Recall that the value of $|\by|$ is distributed as $\Bin(n,1-\delta)$, and given a particular outcome of the value of $|\by|$, the number of bits deleted in Step~5 is distributed as $\min\{0,\bG_1 + \cdots + \bG_{|\by|}-n\}$ where the $\bG_i$'s are independent geometric random variables with each $\bG_i \sim \Geo(\rho)$ (recall that $\rho=1-\delta$).
  We will use two tail bounds: first, by a multiplicative
  Chernoff bound, we have 
  
  \begin{claim} \label{claim:by-not-too-long}
  For $\gamma \leq 1$, we have
  $\Pr[|\by| \ge (1+\gamma) (1-\delta) n] \leq e^{-\Omega(\gamma^2 n)}.$
  \end{claim}
  
  
  The second tail bound shows that $\bG_1 + \cdots + \bG_{|\by|}$ is unlikely to be much larger than $n$:

  \begin{claim} \label{claim:geom-not-too-big}
    Fix an outcome of $|\by|$ such that $|\by| \le (1+\gamma)(1-\delta)n$,
    where $\gamma \leq 1.$
  Then $\Pr[\bG_1 + \cdots + \bG_{|\by|} \geq (1+3 \gamma)n] \leq e^{-\Omega(\gamma^2 n)}.$
  \end{claim}
  
  \begin{proof}
    Recall that \Cref{claim:neg-binomial} upper bounds the probability that $\bG_1 + \cdots + \bG_{(1+\gamma)(1-\delta)n} \geq \frac{1+\gamma}{1-\delta} \cdot (1+\gamma) (1 - \delta) n = (1+\gamma)^2 n$.
  As $\gamma \le 1$, we get that
  \begin{align*}
  \Pr[\bG_1 + \cdots + \bG_{|\by|} \geq (1+3 \gamma)n] 
  \le e^{-\Omega(\gamma^2 |\by|)}
  = e^{-\Omega(\gamma^2 n)},
  \end{align*}
  where the final inequality is by \Cref{claim:neg-binomial}.
  \end{proof}
  
  Combining \Cref{claim:by-not-too-long} and \Cref{claim:geom-not-too-big}, we get that $\bk \leq 3 \gamma n$ except with probability $e^{-\Omega(\gamma^2 n)}$. Consequently, 
we have that
  \begin{align*}
    \E[|\LCS(\wh{\bx},x)]
    &\ge \Bigl( 1 - e^{-\Omega(\gamma^2 n)} \Bigr) \Bigl( \E[|\LCS(\wh{\bx}',x)] - 3\gamma n \Bigr) \\
    &\geq \Bigl( 1 - e^{-\Omega(\gamma^2 n)} \Bigr) \left( 1 - \delta + \frac{\delta^2}{2} - \frac{\delta^3}{2} + \frac{\delta^4}{2} - \frac{\delta^5}{2} \right)n - 3 \gamma n,
  \end{align*}
  and the theorem is proved.
\end{proof}

\medskip

\subsection{Bounds on the performance of one-trace algorithms in the low deletion rate regime}

As noted in the Introduction, it is natural to try to complement \Cref{thm:worst-case-small-delta} by proving an upper bound on the best expected LCS that can be achieved by any one-trace algorithm in the low deletion rate regime. An average-case bound is of course stronger than a worst-case bound of this sort; in \Cref{sec:average-case} we will show that even in the average-case setting, the best achievable LCS given a single trace is at most $(1 - \Theta(\delta))n$.

%% file: average-case-small-rho.tex

\section{Average-case one-trace reconstruction, high and medium deletion rate} \label{sec:average-case-small-rho}

In this section we bound the performance of any average-case few-trace algorithm when the retention rate is low.
Given the length of the source string $n$, we write 
  $L_{0,\avg}(n)$ to denote the performance of an optimal zero trace algorithm:
$$
L_{0,\avg}(n)=\max_{z\in \{0,1\}^n} \Ex_{\bx\sim \{0,1\}^n} \big[\abs{\LCS(\bx,z)}\big]
$$
(note that this quantity does not depend on $\delta$), and recall from \Cref{sec:alg-performance-notation} that for $t>0$,
$L_{t,\avg}(\delta,n)$ captures the information-theoretic optimal performance of any $t$-trace algorithm at deletion rate $\delta$.

We show that when $t\rho$ is small, where $t$ is the number of traces and 
  $\rho$ is the retention rate, it is not possible to do much better than $L_{0,\avg}(n)$:

\begin{theorem} [Average-case upper bound on any algorithm, small retention rate] \label{thm:small-rho-average-case}
Let $n$ be the length of the source string, $t$ be the number of traces and 
  $\rho=1-\delta$ be the retention rate.
Then 
\[
L_{0,\avg}(n) \leq L_{t,\avg}(\delta,n) \leq L_{0,\avg}(n) + t \rho \cdot n.
\]
\end{theorem}

We note that as a special case of the theorem above, if $t(n)\rho(n)=o(n)$ then the leading constant of what can be achieved with $t$ traces is no better than if no traces were given.  This is in contrast with the worst-case setting, as witnessed by \Cref{thm:worst-case-small-rho-informal} and the discussion immediately preceding it.

The lower bound is immediate so in the rest of this section we prove the upper bound.
We start with some notation for working with multiple traces in the proof of Theorem \ref{thm:small-rho-average-case}.
Given $n$ and $t$, let $R_1,\ldots,R_t\subseteq [n]$ be $t$ subsets which should be viewed as
  locations retained from  an $n$-bit string to obtain its $t$ traces (so if the source string is $x$
  then the traces are $y^{(s)}=x_{R_s}$ for $s=1,\dots,t$). 
Let $i_1<\cdots <i_m$ be an enumeration of indices in $ R_1\cup \cdots \cup R_t$.
We write $C=(C_1,\ldots,C_m)$ to denote the tuple where $C_j$ is the set of those $s\in [t]$ such that $i_j\in R_s$.
We will refer to $C$ as the \emph{collision information} of $R_1,\ldots,R_t$, denoted by
  $C=C(R_1,\ldots,R_t)$.

\begin{example}
Consider the case that $n=8$, $t=3$, the source string $x$ is $11010011$, and the three traces $y^{(1)}=1100,$ $y^{(2)}=110,$ $y^{(3)}=1001$ are obtained from $x$ as shown below:
\begin{center}
\begin{tabular}{r c c c c c c c c}
$x:$ & 1 & 1 & 0 & 1 & 0 & 0 & 1 & 1\\
$y^{(1)}:$ & 1 & 1 &  &  & 0 & 0 &  & \\
$y^{(2)}:$ &  & 1 &  & 1 & 0 &  &  & \\
$y^{(3)}:$ &  & 1 &  &  & 0 & 0 & 1 & 
\end{tabular}
\end{center}
In this case we have that $m=6$, $i_1=1$, $i_2 = 2$, $i_3=4$, $i_4=5$, $i_5=6$, $i_6=7$, and $C_1=\{1\}$, $C_2=\{1,2,3\}$, $C_3 = \{2\}$, $C_4 = \{1,2,3\}$, $C_5=\{1,3\}$, $C_6=\{3\}$.  As discussed in \Cref{obs:uniform} below, given the traces $y^{(1)},y^{(2)},y^{(3)}$ and the collision information $C$, it is possible to reconstruct an $m$-bit subsequence $y$ (in this example $y=111001$) of $x$, but not the $n-m$ bits of $x$ that are missing from $y$ nor the locations of where the $m$ bits of $y$ are situated in $x$.
\end{example}

We will consider average-case algorithms that are given not only $t$ traces $\by_1,\ldots,\by_t\sim \Del_\delta(\bx)$
  of a random string $\bx\sim \{0,1\}^n$ but also the collision information $\bC=C(\bR_1,\ldots,\bR_t)$,
  where $\bR_s\subseteq [n]$ is the set of locations that are retained in obtaining trace $\by^{(s)}$ from $\bx$ for each $s\in [t]$.
Let $L_{t,\avg}^*(\delta,n)$ denote the performance of the best algorithm $A$ under this setting:
$$
L_{t,\avg}^*(\delta,n) :=  \max_{A} \Ex_{\bx \sim \zo^n} \Ex_{\bR_1,\dots,\bR_t\sim \calR_\rho}
\Bigl[ \abs[\big]{ \LCS\bigl(A(\bx_{\bR_1},\ldots,\bx_{\bR_t},\bC),\bx\bigr) } \Bigr],
$$
where we write $\calR_\rho$ to denote   the distribution where $\bR\sim \calR_\rho$ is  
  drawn by including each element in $[n]$ independently with probability $\rho$ and
 $\bC=C(\bR_1,\ldots,\bR_t)$ is the collision information of sets $\bR_1,\ldots,\bR_t$.
It is clear that $L_{t,\avg}^*(\delta,n)\ge L_{t,\avg}(\delta,n)$. We prove Theorem \ref{thm:small-rho-average-case} by showing that 
$$
L_{t,\avg}^*(\delta,n)\le L_{0,\avg}(n)+t\rho\cdot n.
$$

\begin{observation}
[\emph{A posteriori} distribution of a uniform random source string given $t$ traces and their collision information] \label{obs:uniform}
Let $\bx$ be a uniform random source string drawn  from $\zo^n$.
Let $I=(y^{(1)},\ldots,y^{(t)},C)$ be any fixed outcome of $t$ traces from $\Del_\delta(\bx)$ together with
  the collision information of their locations retained. Then the \emph{a posteriori} distribution of $\bx$ given $I$ is as follows:
\begin{flushleft}
\begin{enumerate}
\item Let $C=(C_1,\ldots,C_m)$ for some $m\le n$.
We define an $m$-bit string $y$ as follows. For each $j\in [m]$, pick an $s\in C_j$ and set $\smash{z_j=y^{(s)}_k}$ where 
  $k$ is the number of $j'\le j$ such that $s\in C_{j'}$.
  (Note that the value of $z_j$ does not depend on the choice of $s\in C_j$.)
  
\item The rest of the process is the same as the description of 
  the a posteriori distribution of $\bx$ given one trace $y$ (see Observation \ref{ob:post1}); for convenience we will write $\calD_y$ to denote the distribution of $\bx$ described below. Draw a uniform random $m$-element subset of $[n]$ (say $\bS = \{\bs_1,\dots,\bs_m\}$ where $1 \leq \bs_1 < \cdots < \bs_m \leq n$);

\item For each $j \in [m]$ set $\bx_{\bs_j}=z_j$, and for each $i \notin \bS$ set $\bx_i$ to an independent uniform bit.
\end{enumerate}
\end{flushleft}
\end{observation}

We are now ready to prove Theorem \ref{thm:small-rho-average-case}.
\begin{proof}[Proof of Theorem \ref{thm:small-rho-average-case}]
Let $A$ be an optimal algorithm that achieves $L^*_{t,\avg}(\delta,n)$.
Let $\bx\sim \zo^n$~and let $\smash{\bI=(\by^{(1)},\ldots,\by^{(t)},\bC)}$ be the 
  input of $A$, where $A$ outputs $A(\bI)\in \zo^{n}$.
Given an outcome $I$ of $\bI$, we write $y(I)$ to denote the string derived from $\bI$ as in Step 1 of \Cref{obs:uniform}. Then
\begin{align}\label{eq:haha}
L^*_{t,\avg}(\delta,n)=\sum_{I} \Pr\big[\bI=I\big]\cdot \Ex_{\bx\sim \calD_{y(I)}} \Bigl[ \abs[\big]{ \LCS\bigl(A(I),\bx\bigr) } \Bigr],
\end{align}
where the sum is over all possible inputs $I$ of $A$. 
We need the following claim:
\begin{claim}\label{claim:short}
Fix any string $y\in \{0,1\}^m$ for some $m\le n$.
For any string $z\in \{0,1\}^n$, we have 
$$
\Ex_{\bx\sim \calD_y} \Bigl[ \abs[\big]{ \LCS\bigl(z,\bx\bigr) }\Bigr]\le L_{0,\avg}(n)+ m.
$$
\end{claim}
\begin{proof}
Consider the following coupling $(\bx ,\bx')\sim \calE$ of the uniform distribution over $\zo^n$ and $\calD_y$:
  first draw $\bx\sim \zo^n$; then draw a size-$m$ subset $\bS$ of $[n]$ 
  uniformly at random and replace bits of $\bx$ at $\bS$ by $y$ to obtain $\bx'$.
It is easy to verify that $\calE$ is a coupling of the uniform distribution over $\zo^n$ and $\calD_y$.
For any string $z\in \{0,1\}^n$, we have 
\begin{align*}
\Ex_{\bx'\sim\calD_y}\Bigl[ \abs[\big]{ \LCS\bigl( z,\bx'\bigr) } \Bigr]
&=\Ex_{(\bx,\bx')\sim \calE} \Bigl[ \abs[\big]{ \LCS\bigl(z,\bx'\bigr) } \Bigr] \\
&\le \Ex_{(\bx,\bx')\sim \calE}
  \Big[ \abs[\big] {\LCS\bigl(z,\bx\bigr) }\Bigr]+m
  =\Ex_{\bx\sim \zo^n}\Bigl[ \abs[\big] {\LCS\bigl(z,\bx\bigr) }\Bigr]+m\le L_{0,\avg}(n)+m,
\end{align*}
where the inequality used the fact that $(\bx,\bx')\sim \calE$ always have Hamming distance at most $m$.
\end{proof}

Combining (\ref{eq:haha}) with Claim \ref{claim:short}, we have 
$$
L^*_{t,\avg}(\delta,n)\le \sum_{I} \Pr\big[\bI=I\big]\cdot \Big(L_{0,\avg}(n)+|y|\Big)
=L_{0,\avg}(n)+\bE\big[|\by|\big].
$$
By linearity of expectation, we have 
$$
\bE\big[|\by|\big]=
n (1 - \delta^t) = n(1-(1-\rho)^t) \leq \rho t \cdot n.
$$
This finishes the proof of the theorem.
\end{proof}


%% file: average-case-small-delta.tex

\def\cs{\mathsf{CS}}
\section{Average-case one-trace reconstruction, small deletion rate} \label{sec:worst-case-small-rho-informal}

\subsection{An efficient algorithm  improving on the \Cref{thm:worst-case-small-delta} bound}

In this section we show that the algorithm \textsc{Small-rate-reconstruct} of \Cref{thm:worst-case-small-delta}, that was shown to achieve LCS $(1 - \delta + \delta^2/2 - \delta^3/2 + \delta^4/2 - \delta^5/2 - o(1))n$ for worst-case source strings, in fact does better than this for average-case source strings.
The high level idea is that when there are $j$ additional bits between two trace bits in $\bx$ and $k$ additional bits between two trace bits in $\wh{\bx}'$, rather than matching only $\min\{j,k\}/2$ using the randomness of $\wh{\bx}'$, we take advantage of the facts that (i) both the $\bx$-bits and the $\wh{\bx}$-bits are uniform random, and (ii) if $j$ or $k$ is greater than 1, then the expected LCS between a random $j$-bit string and a random $k$-bit string is strictly larger than $\min\{j,k\}/2$, to obtain (on average) a better matching between these two blocks and hence a larger overall matching. This intuition motivates the following definition:
\begin{definition}
For integers $j,k >0$, we define $\cs(j,k)$ as 
\[\cs(j,k) := \Ex_{\bx \sim \zo^j,\bx' \sim \zo^k} \bigl[ |\LCS(\bx,\bx')| \bigr]. \]
Note that by definition, $\cs(j,k) = \cs(k,j)$, i.e., the function $\cs(\cdot, \cdot)$ is symmetric in its arguments. 
\end{definition}
While it is not clear if there is a simple explicit formula for $\cs(j,k)$, we note that a brute force algorithm can be used to compute this function. Further, for the special case of $j=k$, the function $\cs(\cdot, \cdot)$ has been studied previously in the literature~\cite{ChvatalSankoff75}. In particular, for any $d>0$, $\cs(d,d)$ is the same as the function $f(d,2)$ defined in \cite[Section 2]{ChvatalSankoff75}. Further, once $d \rightarrow \infty$, $\cs(d,d)/d$ is the same as the so-called Chvatal--Sankoff constant for the binary alphabet~\cite{kiwi2005expected,  ChvatalSankoff75}. \Cref{Tab:Tcr} gives the values of $\cs(j,k)$ for all $j + k \le 6$. 
\begin{center}
\begin{table}
\centering
\begin{tabular}{|*{6}{c|}}\cline{1-2}
  $k=5$ & $31/32$                 \\ \cline{1-3}
  $k=4$ & $15/16$ & $53/32$             \\ \cline{1-4}
  $k=3$ & $7/8$ & $23/16$ & $29/16$        \\ \cline{1-5}
  $k=2$ & $3/4$ & $9/8$ & $23/16$ & $53/32$    \\ \hline
  $k=1$ & $1/2$ & $3/4$ & $7/8$ & $15/16$ & $31/32$  \\ \hline
 value of $\cs(j,k)$   & $j=1$ & $j=2$ & $j=3$ & $j=4$ & $j=5$  \\ \hline
\end{tabular}
  
\caption{Table for values of $\cs(j,k)$ for $j+k \le 6$.} 
\label{Tab:Tcr}
\end{table}
\end{center}
\begin{theorem} \label{thm:small-delta-average-case-algorithm}
Let $\delta=\delta(n)$ be the deletion rate. The $O(n)$-time algorithm \textsc{Small-rate-reconstruct} given in \Cref{thm:worst-case-small-delta} has the following property:  for any $\gamma > 0$ and sufficiently large $n$, algorithm \textsc{Small-rate-reconstruct} outputs a hypothesis string $\wh{\bx} \in \zo^n$ satisfying
  \begin{align*}
&\Ex_{\bx \in \zo^n}\Ex_{\by \sim \Del_{\delta}(\bx)} \bigl[ |\LCS(\wh{\bx},\bx)| \bigr] \\ 
&\geq
  \Bigl(1 - e^{-\Omega(\gamma^2 n)} \Bigr) (1 - \delta) \cdot \biggl( 1 + (1-\delta)^2\sum_{j,k=1}^\infty  \cs(j,k) \cdot \delta^{j+k}  \biggr) n  - 3\gamma n. 
  \end{align*}
  \end{theorem}
As an example,  instantiating with the values of $\cs(j,k)$ from Table~\ref{Tab:Tcr}, the above theorem  gives us that
\[
  L_{1,\avg}(\delta,n) 
  \ge \left( 1 - \delta + \frac12\delta^2   + \frac{17}{8}\delta^4 + \frac{55}{8}  \delta^5 + o(\delta^5) \right) n ,
\]
which improves on the  $(1 - \delta + \delta^2/2 -\delta^3/2   + \delta^4/2 -   \delta^5/2 - o(1))n$ bound of \Cref{thm:worst-case-small-delta}.
\ignore{
\gray{
So, for example, \Cref{thm:small-delta-average-case-algorithm} gives us that
\[
  L_{1,\avg}(\delta,n) 
  \ge \left( 1 - \delta + \frac{\delta^2}{2} - \frac{\delta^3}{2} + \frac{5\delta^4}{8} - \frac{5\delta^5}{8} + o(\delta^5) \right) n ,
\]
which improves on the  $(1 - \delta + \delta^2/2 - \delta^3/2 + \delta^4/2 -  \delta^5/2 - o(1))n$ bound of \Cref{thm:worst-case-small-delta}.
}}

\begin{proof}
  The proof is analogous to the one in \Cref{thm:worst-case-small-delta}.
  We first replace $x$ in the proof of \Cref{thm:worst-case-small-delta} with a uniform random $\bx \sim \zo^n$ in the proof.

  Now, for each $i \in [\abs{\by}]$, let us define $\bd_i$ to be $\abs{\bx^i}$ and $\bd'_i$ to be  $\abs{\wh{\bx}'^i}$. Since the length-$(\bd_i-1)$ prefix of both $\bx^i$ and the length-$(\bd'_i-1)$ prefix of $\wh{\bx}'^i$ are independent random strings, and the last bit of both $\bx^i$ and $\wh{\bx}'^i$ are the same, we have  
  \[
    \E_{\bx^i, \wh{\bx}'^i} \Bigl[ \abs{\LCS(\bx^i, \wh{\bx}^i)} \Bigr] \ge \cs(\bd_i-1,\bd'_i-1) + 1.
  \]
  As $\bd_i \sim \Geo(1-\delta)$ and $\bd'_i \sim \Geo(1-\delta)$, we have
  \begin{align*}
    \E\bigl[\cs(\bd_i-1,\bd'_i-1)\bigr]
    &= \sum_{j,k=1}^\infty \Bigl( \cs(j,k) \cdot \Pr[\bd_i = j+1 \text{~and~}\bd'_i = k+1] \Bigr) \\
    &= \sum_{j,k=1}^\infty \Bigl( \cs(j,k) \cdot \Pr[\Geo(1-\delta) = j+1] \cdot \Pr[\Geo(1-\delta) = k+1] \Bigr) \\
    &=  (1-\delta)^2\sum_{j,k=1}^\infty  \cs(j,k) \cdot \delta^{j+k} .
  \end{align*}
  So we have
  \begin{align*}
    \E \Bigl[ \abs{\LCS(\bx, \wh{\bx}')} \Bigr]
    &\ge \sum_{i=1}^{\abs{\by}} \E_{\bx^i, \wh{\bx}'^i}  \Bigl[ \abs{\LCS(\bx^i, \wh{\bx}'^i)} \Bigr] \\
    &\ge \E\bigl[ \abs{\by} \bigr] + \E \bigl[ \abs{\by} \bigr] (1-\delta)^2\sum_{j,k=1}^\infty  \cs(j,k) \cdot \delta^{j+k} \\
    &\ge (1 - \delta) n \cdot \biggl( 1 + (1-\delta)^2\sum_{j,k=1}^\infty  \cs(j,k) \cdot \delta^{j+k}\biggr).
  \end{align*}
  We can again relate $\E[\abs{\LCS(\wh{\bx}, \bx)}]$ to $\E[\abs{\LCS(\wh{\bx}', \bx)}]$ using the same argument in \Cref{thm:worst-case-small-delta}, from which we conclude that
  \begin{align*}
    \E\Bigl[ \abs{\LCS(\wh{\bx}, \bx)} \Bigr]
    &\ge \Bigl( 1 - e^{-\Omega(\gamma^2 n)} \Bigr) \Bigl( \E\Bigl[ \abs{\LCS(\wh{\bx}, \bx')} \Bigr] - 3 \gamma n \Bigr) \\
    &\ge \Bigl( 1 - e^{-\Omega(\gamma^2 n)} \Bigr) 
   (1 - \delta) n \cdot \biggl( 1 + (1-\delta)^2\sum_{j,k=1}^\infty  \cs(j,k) \cdot \delta^{j+k}\biggr)
- 3\gamma n ,
  \end{align*}
  proving the theorem. 
  \end{proof}

\subsection{Bounds on the performance of any one-trace algorithm}

Finally, in this section we establish an upper bound on the best possible performance that any one-trace algorithm can achieve in the small-deletion-rate regime.
We consider the average-case setting (which is of course more challenging for upper bounds, and yields worst-case upper bounds as an immediate consequence).

A relatively simple analysis shows that 
 $L_{1,\avg}(\delta,n) \leq (1 - c\delta/\log(1/\delta))n$, where $c$ is a universal positive constant.  This argument applies a union bound across all possible matchings of a given size, and is given as \Cref{thm:small-delta-average-case-upper-bound-weak} in~\Cref{sec:small-delta-average-case-weak-upper-bound}.  It is natural to suspect that this bound is weaker than it should be by a $\Theta(\log(1/\delta))$ factor, but establishing this turns out to be nontrivial.  The following theorem establishes a bound of $L_{1,\avg}(\delta,n) \leq (1 - c \delta)n$, which, up to the value of the universal constant $c$, is best possible by \Cref{thm:small-delta-average-case-algorithm} (or even by the trivial algorithm which simply outputs any $n$-bit string that contains the input trace $\by$, and thereby achieves an LCS of expected length at least $\E[|\by|]=(1-\delta)n$).

\begin{theorem}
[Average-case upper bound on any algorithm, small retention rate] 
\label{thm:small-delta-average-case-upper-bound}
There is an absolute constant $c>0$ such that for any  deletion rate $\delta=\delta(n)=\omega(1/n)$ and  sufficiently large $n$, we have $L_{1,\avg}(\delta,n) \leq (1 - c \delta)n.$
\end{theorem}

\def\OPT{\textsf{OPT}}

\subsubsection{Outline of the argument} 

Recall that under the average-case setting, the source string $\bx$ is uniform random over $\zo^n$, and our goal is to upperbound the performance of any algorithm which is given as input a single trace $\by \sim \Del_\delta(\bx).$
Given any trace string $y$, the optimal algorithm $A$ will return a string $z\in \zo^n$ to maximize
  $\E_{\bx \sim y}[|\LCS(z,\bx)|]$ (recall Observation \ref{ob:post1} for the 
    a posteriori distribution $\bx\sim y$).
Let us write $\opt(y)$ for an optimal string $z\in \zo^n$ for the expectation.
Given that $\by \sim \Del_\delta(\bx)$ and $\bx\sim \zo^n$,
  our goal is to bound $$
L_{1,\avg}(\delta,n)=\E_{\by}\Big[\E_{\bx\sim \by}\big[|\LCS(\opt(\by),\bx)|\big]\Big],
$$
where $\by$ 
  is a uniform random bitstring of length $\bk$ where $\bk \sim \Bin(n,1-\delta)$. 
For each $k$, let
$$
\OPT_k:=\E_{\by\sim \zo^k} \Big[\E_{\bx\sim \by}\big[|\LCS(\opt(\by),\bx)|\big]\Big].
$$
It is easy to see that $\OPT_k$ is nondecreasing in $k$,\footnote{To see this, note that
  any algorithm that receives a random trace of length $k$ can be simulated by 
  an algorithm that receives a trace of length $k+1$ by randomly deleting one bit from its input trace.}
  and we have (with $\bk\sim \Bin(n,1-\delta)$)
\begin{equation}\label{eq:blabla1}
L_{1,\avg}(\delta,n)=\sum_{k=0}^n \Pr\big[\bk=k]\cdot \OPT_k.
\end{equation}
 
Let $\delta'=\delta/2$ and $m=(1-\delta')m$. 
To prove \Cref{thm:small-delta-average-case-upper-bound}, we first show that it suffices
  to obtain the following upper bound for $\OPT_m$: 
\begin{equation}\label{eq:blabla2}
\OPT_m\le (1-c_1\delta')n,
\end{equation}  
for some universal positive constant $c_1$.
Consequently it suffices to analyze the optimal one-trace algorithm which is given as input a uniform random string $\by \sim \zo^m$. 

Next, we observe that by a simple triangle inequality argument, it suffices to show that 
\begin{equation}
\label{eq:jarndyce}
\Ex_{\by \sim \zo^m} \Ex_{\bx, \bx' \sim \by}\big[|\LCS(\bx,\bx')|\big]\le (1-2c_1\delta')n
\end{equation}
for some constant $c_1$;
in turn, to prove (\ref{eq:jarndyce}), it is enough to bound (informally)
\begin{equation}
\label{eq:summerson}
\Prx_{\by \sim \zo^m, \bx,\bx' \sim \by}\big[\text{$\bx$ and $\bx'$ have a ``large'' matching}\big].
\end{equation} 

In
\Cref{claim:score} we give an upper bound on the probability that a \emph{fixed} large candidate matching $M$ is a valid matching between $\bx,\bx' \sim \by.$ 
This upper bound is in terms of a quantity that we call $\score_M$, which depends on the candidate matching $M$ and is a random variable whose randomness comes from the sets $\bS$ and $\bS'$ as in \Cref{obs:uniform}'s description of the distribution of $\bx\sim \by$ and $\bx' \sim \by$. 
As we show in \Cref{claim:tailhighscore}, to establish \Cref{eq:summerson} it is enough to show that for \emph{every} large candidate matching $M$, an upper tail bound on $\score_M(\bS,\bS')$ holds.  We prove such a tail bound in \Cref{lem:highscoreunlikely}. The two main steps are (i) showing (in \Cref{claim:well-spaced}) that $\Pr[\hspace{0.05cm}\bS$ is not ``\emph{well-spaced}''$\hspace{0.02cm}]$ is very small (see \Cref{def:well-spaced} for the definition of well-spaced sets), and (ii) showing (in \Cref{claim:scorelargesmallprob}) that if $\bS=S$ is good, then $\Pr_{\bS'}[\hspace{0.05cm}\score(S,\bS')$ is large$\hspace{0.04cm}]$ is very small.

\subsubsection{Proof of \Cref{thm:small-delta-average-case-upper-bound}}

We may assume that $\delta=\delta(n)$ is at most some sufficiently small universal positive constant, since otherwise the claimed bound follows immediately from \Cref{thm:small-delta-average-case-upper-bound-weak}. We will use this assumption in various bounds throughout the proof.

Let $\delta'=\delta/2$ and $m=(1-\delta')n$ (so $\delta'$ is $\omega(1/n)$ and at most 
  some sufficiently small universal positive constant as well).
By the well-known fact \cite{KaasBuhrman80} that the median of the $\Bin(n,\delta)$ distribution belongs to $\{\lfloor n \delta \rfloor,\lceil n \delta \rceil\}$ (which is at least $\delta'n$ using $\delta=\omega(1/n)$),
  it follows from (\ref{eq:blabla1}) and the monotonicity of $\OPT_k$ that 
$$
L_{1,\avg}(\delta,n)\le 0.5\cdot \OPT_m + 0.5n
$$
and thus, to prove \Cref{thm:small-delta-average-case-upper-bound} 
  it suffices to obtain the upper bound for $\OPT_m$ in (\ref{eq:blabla2}).



Instead of working with $\OPT_m$ directly,
  an application of the triangle inequality lets us work with the expression on the LHS of
  (\ref{eq:jarndyce}) which (conveniently) does not involve $\opt(\by)$:
\begin{claim} \label{claim:peepy}
Suppose that \Cref{eq:jarndyce} holds.
Then \Cref{eq:blabla2} holds.
\end{claim}
\begin{proof}
For any $y\in \zo^m$, any $n$-bit string $\opt(y)$, and any two $n$-bit strings $x,x'$, we have that the length of the $\LCS$ between $x$ and $x'$ is at least the number of coordinates of $\opt(y)$ that participate both in the optimal matching between $\opt(y)$ and $x$ and in the optimal matching between $\opt(y)$ and $x'$. Since this number is at least
$|\LCS(x,\opt(y))|+|\LCS(x',\opt(y))| - n$, we have that
\begin{equation}
n + \big|\LCS(x,x')\big| \geq \big|\LCS(x,\opt(y))\big| + \big|\LCS(x',\opt(y))\big|. \label{eq:homais}
\end{equation}
It follows that
\begin{align*}
2(1 - c_1 \delta') n &= n + (1 - 2c_1 \delta')n\\[0.5ex]
&\geq
n + \Ex_{\by \sim \zo^m} \Ex_{\bx, \bx' \sim \by}\big[|\LCS(\bx,\bx')|\big]\\[0.6ex]
& \geq \Ex_{\by \sim \zo^m} \Ex_{\bx , \bx' \sim \by}\big[|\LCS(\bx,\opt(\by)) | + |\LCS(\bx',\opt(\by))|\big]\\[0.8ex]
&= 2 \Ex_{\by \sim \zo^m} \Ex_{\bx \sim \by}\big[|\LCS(\bx,\opt(\by))|\big],
\end{align*}
where the first inequality is by \Cref{eq:jarndyce},  the second is by \Cref{eq:homais} (averaged over $\by$, $\bx \sim \by$ and $\bx' \sim \by$), and the third is because $\bx'$ and $\bx$ are identically distributed.
\end{proof}

Given \Cref{claim:peepy}, our goal in the rest of the proof is to establish \Cref{eq:jarndyce}.
We note that in \Cref{eq:jarndyce}, given the outcome of $\by$, the two $n$-bit strings $\bx$ and $\bx'$ are \emph{independently} distributed according to $\bx \sim \by$ and $\bx' \sim \by$; in particular, recalling \Cref{obs:uniform}, there are two independent draws performed to obtain the sets $\bS$ (for $\bx$) and $\bS'$ (for $\bx'$). This independence will be used heavily in the rest of the argument.

Recalling \Cref{obs:uniform}, we rewrite \Cref{eq:jarndyce} as
\begin{equation} \label{eq:realgoal2}
\Ex_{\by,\bS,\bS',\br,\br'}\big[|\LCS(\bx,\bx')|\big] \leq \left(1 - 2c_1 \delta'\right)n,
\end{equation}
where $\by \sim \zo^{m}$, $\bS$ and $\bS'$ are independent uniform $m$-element subsets of $[n]$,  and $\br,\br'$ are independent uniform draws from $\zo^{n-m}$ representing the ``rest of the bits'' that get filled into the locations in $[n] \setminus \bS$ and $[n] \setminus \bS'$ to complete the $n$-bit strings $\bx$ and $\bx'$, respectively. Recall that $\bx$ has $\by$ in the $m$ locations of $\bS$ and $\br$ in the other $n-m$ locations, and $\bx'$ gets the same $\by$ in the locations of $\bS'$ and $\br'$ in the other locations.

Since the length of the $\LCS$ between two strings is the size of the largest matching between them, to establish \Cref{eq:realgoal2} it suffices to prove that
\begin{equation}
\label{eq:realgoal3}
\Prx_{\by,\bS,\bS',\br,\br'}\big[\hspace{0.05cm}\text{there exists a matching between $\bx$ and $\bx'$ of size $(1- 4c_1\delta')n$}\hspace{0.04cm}\big] \leq 1/2 
\end{equation}
for some universal positive constant $c_1$.
Thus our remaining task is to establish \Cref{eq:realgoal3}.

\subsubsection{Matchings and scores}

Recall from \Cref{sec:preliminaries} that a matching $M$ of size $t$  between two $n$-bit strings $z,z'$ is a sequence of pairs $M=(M_1,\dots,M_t)$, where 
\begin{enumerate}

\item [(a)] $M_i=(v_i,v'_i)$ are such that $1 \leq v_1 < \cdots < v_t \leq n$, $1 \leq v'_1 < \cdots < v'_t \leq n$, and 
\item [(b)] for each $i \in [t]$ we have that the two bits $z_{v_i}$ and $z'_{v'_i}$ are the same.

\end{enumerate}
Let us say that a \emph{candidate matching} is a sequence of pairs $M=(M_1,\dots,M_t)$ satisfying (a); if moreover (b) holds for a pair of $n$-bit strings $z$ and $z'$, we say that the candidate matching $M$ is \emph{valid for $(z,z')$}.

Let $M=(M_1,\dots,M_t)$ be a candidate matching and let $$S=\{s_1 < \cdots < s_m\}\quad\text{and}\quad S'=\{s'_1 < \cdots < s'_m\}$$ be two $m$-element subsets of $[n].$ We say that an edge $M_i = (v_i,v'_i)$ of $M$ \emph{synchs up} with the pair $(S,S')$ if there is some $j \in [m]$ such that $v_i=s_j$ and $v'_i=s'_j$; in words, for some $j$ the candidate matching attempts to match up the $j$-th element of $S$ with the $j$-th element of $S'$. We say the \emph{score of $M$ on $(S,S')$}, denoted $\score_M(S,S')$, is
\[
\score_M(S,S'):=
\Big|\big\{i \in [t] : M_i \text{~synchs up with~}(S,S')\big\}\Big|,
\]
the number of edges of $M$ that match up corresponding elements of $S$ and $S'$.

\begin{claim} \label{claim:score}
Let $M=(M_1,\dots,M_t)$ be a candidate matching of size $t$ and let $S,S'$ be $m$-element subsets of $[n]$ such that $\score_M(S,S')=\ell$. Then
\[
\Prx_{\by,\br,\br'}\big[\hspace{0.03cm}M \text{~is a valid matching for~}(\bx,\bx')\hspace{0.03cm}\big] = 
{\frac 1 {2^{t-\ell}}},
\]
where $\bx$ and $\bx'$ are defined based on $S,S',\by,\br,\br'$ as described after \Cref{eq:realgoal2}.
\end{claim}
\begin{proof}
For each $M_i=(v_i,v'_i)$ that synchs up with $(S,S'),$ it is clear that $\smash{\bx_{v_i}=\bx'_{v'_i}}$, because both  are the same bit $\by_j$ of the string $\by$.  
There are $t-\ell$ remaining equalities $$\bx_{v_{i_1}} \stackrel{?}{=}\bx_{v'_{i_1}}, \dots, \bx_{v_{i_{t-\ell}}} \stackrel{?}{=}\bx_{v'_{i_{t-\ell}}},\quad\text{where $i_1<\cdots<i_{t-\ell}$,}$$ that must all hold in order for the candidate matching $M$ to be valid for $(\bx,\bx')$, corresponding to the $t-\ell$ edges of $M$ that do not synch up with $(S,S').$  Each of these equalities holds independently with probability $1/2$. This can be seen by considering the $t-\ell$ edges of $M$  successively in increasing order ``from left to right'': for each $j \in [t-\ell]$, for any given outcome of the bits of $\by,\br$ and $\br'$ that were involved in the first $j-1$ edges, there is a ``fresh random bit'' from either $\by, \br$ or $\br'$ involved in the $j$-th edge that causes the $j$-th equality to hold with probability 1/2.
\end{proof}

For the rest of the proof of \Cref{thm:small-delta-average-case-upper-bound},
 we fix $t:=(1-4c_1\delta')n$ for some universal constant $c_1$ to be picked later  (recall that that is the size of the matchings that we are concerned with in \Cref{eq:realgoal3}).  The following claim states that it suffices to establish that for each fixed size-$t$ candidate matching $M$, the probability that it has a high score is very low:

\begin{claim} \label{claim:tailhighscore}
Suppose that there is a universal positive constant $c_1$ such that 
  the following inequality holds for each candidate matching $M$ of size  $t=(1-4c_1\delta')n$:
\begin{equation} \label{eq:highscoreunlikely}
\Prx_{\bS,\bS' \sim {[n] \choose [m]}}
\Big[\hspace{0.05cm}\score_M(\bS,\bS') \geq \big(1-3H(4c_1 \delta')\big)n\hspace{0.05cm}\Big]
  \le {\frac 1 {4 \cdot 2^{2H(4c_1 \delta')n}}}.
\end{equation}
Then \Cref{eq:realgoal3} holds with the same constant $c_1$.
\end{claim}
\begin{proof}
We use $\Pr[\bA] \leq \Pr[\bB] + \Pr[\bA \ | \ \overline{\bB}]$ where $\bA$ is the event ``there exists some valid matching between $\bx$ and $\bx'$ of size $t$''
and $\bB$ is the event ``there exists some candidate matching $M$ of size $t$ with $\score_M(\bS,\bS') \geq (1-3H(4c_1 \delta'))n$.''
There are $${n \choose 4c_1\delta' n}^2 \leq2^{2H(4c_1 \delta')n}$$ 
  many candidate matchings $M$ of size $t$. A union bound together with \Cref{eq:highscoreunlikely} gives that $$\Pr[\bB] \leq 
2^{2H(4c_1 \delta')n} \cdot 
 {\frac 1 {4 \cdot 2^{2H(4c_1 \delta')n}}}
= {\frac 1 4}.$$

To upperbound $\Pr[\bA \ | \ \overline{\bB}]$, 
fix any particular outcome $(S,S')$ of $(\bS,\bS')$ such that $\score_M(S,S')$ $< (1-3H(4c_1 \delta'))n$ holds for every candidate matching $M$ of size $t$. 
By \Cref{claim:score} we have that 
\begin{align*}
\Prx_{\by,\br,\br'}\big[M \text{~is valid for~}(\bx,\bx')\ | \ (\bS,\bS')=(S,S')\big] &\leq
{\frac 1 {2^{t - (1-3H(4c_1 \delta'))n}}} 
 <
{\frac 1 {2^{2.5H(4c_1 \delta')n}}},
\end{align*}
where the second inequality uses that $\delta'$ is at most some sufficiently small absolute constant.
By a union bound over all (at most $2^{2H(4c_1 \delta')n}$) many candidate matchings $M$ of size $t$, we see that 
$$\Pr\big[\bA \ | \ (\bS,\bS') = (S,S')\big]\le 
 2^{2H(4c_1 \delta')n} \cdot 
{\frac 1 {2^{2.5H(4c_1 \delta')n}}} \leq {\frac 1 4},$$ where the inequality holds since $\delta'=\omega(1/n)$.
Hence $\Pr[\bA \ | \ \overline{\bB}]\le 1/4$, and the claim is proved.
\end{proof}

For the rest of the proof fix $M$ to be any particular size-$t$ candidate matching. By \Cref{claim:tailhighscore}, our remaining task is to establish the tail bound on $\score_M(\bS,\bS')$ that is asserted by \Cref{eq:highscoreunlikely}. Since $\delta'$ is at most some absolute constant, this is an immediate consequence of the following slightly stronger (and cleaner to state) version:

\begin{lemma}
\label{lem:highscoreunlikely}
There is a universal positive constant $c_1$ such that
\begin{equation} \label{eq:highscoreunlikely2}
\Prx_{\bS,\bS' \sim {[n] \choose [m]}}
\Big[\hspace{0.05cm}\score_M(\bS,\bS') \geq \big(1-\sqrt{\delta'}\big)n\hspace{0.05cm}\Big]
\leq {\frac 1 {4 \cdot 2^{2H(4c_1 \delta')n}}}.
\end{equation}
\end{lemma}

\subsubsection{Proof of \Cref{lem:highscoreunlikely}}

Let the size-$t$ matching $M$ be given by $M=((v_1,v'_1),\dots,(v_t,v'_t))$. 
We define sets $L := \{v_1,\dots,v_t\}$ and $R:=\{v'_1,\dots,v'_t\}$
  with $v_1<\cdots<v_t$ and $v_1'<\cdots<v_t'$.
  
In the proof of \Cref{lem:highscoreunlikely} it will be sometimes convenient for us to view $\bS$ as a uniform random string from $\zo^n$ conditioned on containing exactly $m$ ones (and $\bS'$ as an independent random string with the same distribution). We write $\zo^n_m$ to denote the set of all such $n$-bit strings with exactly $m$ ones.
 
The key notion for the proof of \Cref{lem:highscoreunlikely} is the following:

\begin{definition} \label{def:well-spaced}
We say that an outcome $S \in \zo^n_m$ of the random variable $\bS$ is \emph{well-spaced} if it has the following property: there are at least $\delta' n/2$ many disjoint intervals $I_1,\dots,I_{\delta' n/2} \subset [n],$ each of length exactly $1+2\beta$ with 
$\smash{ \beta:= 1/\delta'^{3/4} }$, such that for each $j \in [\delta' n /2]$ we have that
\begin{flushleft}
\begin{itemize}
\item [(i)] $I_j$ is entirely contained in $L$ (so $I_j$ contains $v_{i_j},\ldots,v_{i_j+2\beta}$ for 
  some $i_j$) and moreover, their corresponding indices in $R$ ($v_{i_j}',\ldots,v_{i_j+2\beta}'$)
  also form an interval (i.e., $v_{i_j+2\beta}'=v_{i_j}'+2\beta$); and

\item [(ii)] viewing $S$ as a bit-string from $\{0,1\}^n_m$, the subword $S_{I_j}$ of $S$ is $1^\beta 0 1^\beta$,  i.e. there is a 0 exactly in the middle of interval $I_j$ and the other $2\beta$ bits in the interval are all 1.  
\end{itemize}
\end{flushleft}
\end{definition}

Given \Cref{def:well-spaced}, \Cref{lem:highscoreunlikely} is an immediate consequence of \Cref{claim:well-spaced} and \Cref{claim:scorelargesmallprob} using $\Pr[\bA] \leq \Pr[\bB] + \Pr[\bA \ | \ \overline{\bB}]$ where $\bA$ is the event ``$\score_M(\bS,\bS') \geq (1-\sqrt{\delta'})n$'' and $\bB$ is the event ``$\bS$ is not well-spaced,''
  and taking $c_1$ to be a suitably small constant relative to those constants hidden in
  the $\Omega(\cdot)$ of these two claims. 

\begin{claim} \label{claim:well-spaced}
We have $$\Prx_{\bS \sim \zo^n_m}\big[\hspace{0.05cm}\text{$\bS$ is not well-spaced}\hspace{0.05cm}\big] \leq 2^{-\Omega(\delta' \log(1/\delta')n)}.$$
\end{claim}

\begin{claim}
\label{claim:scorelargesmallprob}
Fix any well-spaced $S \in \zo^n_m.$ Then
$$
\Prx_{\bS' \sim \zo^n_m}\Big[\hspace{0.05cm}\score_{M }(S,\bS') \geq \big(1-\sqrt{\delta'}\big)n\hspace{0.05cm}\Big] \leq 2^{-\Omega(\delta' \log(1/\delta'))n}.
$$
\end{claim}

\begin{proof}[Proof of \Cref{claim:well-spaced}]
We view the draw of $\bS$ as a sequential process in which the outcomes of different groups of coordinates are successively revealed.
We first reveal the outcome of $\bS_{[n]\setminus L}=S_{[n]\setminus L}$, and we consider the remaining distribution over the outcome of $\bS_L.$
Let $b$ be the number of $0$'s in $S_{[n]\setminus L}$. Then
  the remaining distribution of $\bS_L$ is uniform random over all strings in $\smash{\zo^L}$ 
  that contains exactly $a:=n-m-b$ many zeros.
Given that $b\le |[n]\setminus L|=4c_1\delta 'n\le 0.01\delta' n$ (using $c_1\le 1/400$)  
  we have $a \in [0.99\delta' n,\delta' n]$.

After $\bS_{[n]\setminus L}$ is drawn,
 we can view a draw of $\bS_{L}$ from the above-described distribution as being obtained through a sequential random process, proceeding for $a$ stages, where in the $j$-th stage, after locations $\bi_1,\dots,\bi_{j-1}$ in $[t]$ for zeros have been selected in the first $j-1$ stages, a new uniform random location $\bi_j$ in $[t] \setminus \{\bi_1,\dots,\bi_{j-1}\}$ is selected for the $j$-th zero (which means that the $v_{\bi_j}$-th entry of $\bS$ is set to zero). 
After each stage we keep track of the number of locations $i\in [t]$ selected 
  so far such that 
\begin{enumerate}
\item none of $i-2\beta,\ldots,i-1,i+1,\ldots,i+2\beta$ was selected so far; and
\item both $v_{i-\beta},\ldots,v_{i+\beta}$ and $v_{i-\beta}',\ldots,v_{i+\beta}'$ form an interval of length $2\beta+1$. 
\end{enumerate}
We write $\bX_j$ to denote this random variable after $j$ stages.
It suffices to show that $\bX_a$, after all $a$ stages, is at least $\delta'n/2$ with high probability.
 
To this end, we first notice that after the $j$-th stage, the number $\bX_j$ can go down from $\bX_{j-1}$ by at most two.
On the other hand, it goes up by one when $\bi_j$ is not one of the following ``disallowed'' locations $i\in [t]$:
\begin{flushleft}\begin{enumerate}
\item $v_{i-\beta},\ldots,v_{i+\beta}$ or $v_{i-\beta}',\ldots,v_{i+\beta}'$ 
  does not form an interval; the number of such $i\in [t]$ is at most 
  $2\cdot 2\beta\cdot (n-t).$
\item $i$ is within $2\beta$ of a location already picked; the number of such $i$ is 
  at most $(4\beta+1)a$.
\end{enumerate} \end{flushleft}
As a result, the probability that $\bX_j$ does not go up by one is at most
$$
\frac{4\beta(n-t)+(4\beta+1)a}{t-(j-1)}=\frac{16c_1\delta'^{1/4}n+(4\beta+1)a}{t-(j-1)}\le 5\delta'^{1/4},
$$
where we used $a\le \delta'n$. 
%
Consequently, the probability that out of the $a$ stages in which a location is chosen, at least 
  $\delta'n/10$ times it does not go up by one is at most
\begin{align*}
2^a \cdot (5\delta'^{1/4})^{ \delta' n / 10}\le 2^{\delta'n}\cdot (5\delta'^{1/4})^{ \delta' n / 10} 
= 2^{-\Omega(\delta' \log(1/\delta')n)},
\end{align*}
when $\delta'$ is sufficiently small.
If this does not happen, then $\bX_a$ at the end is at least $$a-(\delta'n/10)-2\cdot (\delta'n/10)\ge \delta'n/2$$ using $a\ge 0.99\delta'n$, so the claim is proved.
\end{proof}

\begin{proof}[Proof of \Cref{claim:scorelargesmallprob}]
Let $S$ be a well-spaced set in $\{0,1\}^n_m$, and 
  $I_1,\ldots,I_{\delta'n/2}$  be the $\delta'n/2$ intervals in $[n]$ of length $2\beta+1$ each  that satisfy the conditions of \Cref{def:well-spaced}.
For each $I_j$, we let $i_j\in [t]$ be such that $I_j=\{v_{i_j-\beta},\ldots,v_{i_j+\beta}\}$
  (so $v_{i_j}$ is the center of $I_j$).
Let $$I_j' 
=\left\{v'_{i_j-4/\sqrt{\delta'}},\ldots,v'_{i_j+4/\sqrt{\delta'}}\right\}
$$
for each $j$.
Note that since $\delta'$ is at most some sufficiently small constant, we have 
  $4/\sqrt{\delta'}<\beta$ and thus,
  $I_j'$'s are mutually disjoint intervals in $[n]$ because $I_j$'s satisfy
  conditions of \Cref{def:well-spaced}.

Let $\bS'\sim \zo^n_m$.
We claim that, 
in order to have $\smash{\score_{M}(S,\bS') \geq (1-\sqrt{\delta'})n}$, it must be the case 
  that $\bS'$ has at least $\delta'n/4$ many zero entries in the union of $I_j'$.
To see this, suppose that $\bS'$ has no more than $\delta'n/4$ many zeros in the union of $I_j'$.
Then at least $\delta'n/4$ many $I_j'$'s have all ones in $\bS'$.
For each such $j$, given that $S_{I_j}=1^\beta 0 1^\beta$, we have that either 
$$
\left(v_{i_j-4/\sqrt{\delta'}},v'_{i_j-4/\sqrt{\delta'}}\right),\ldots,\left(v_{i_j-1},v'_{i_j-1}\right)\quad\text{or}\quad
 \left(v_{i_j+1},v'_{i_j+1}\right),\ldots,\left(v_{i_j+4/\sqrt{\delta'}},v'_{i_j+4/\sqrt{\delta'}}\right) 
$$
are not synched. As a result, the number of pairs in $M$ that are not synched in $(S,\bS')$
  is at least $(\delta'n/4)\cdot (4/\sqrt{\delta'})\ge \sqrt{\delta'}n$ and thus, the score is at most $(1-\sqrt{\delta'})n$.

Finally we bound the probability of $\bS'\sim \zo^n_m$ having at least $\delta'n/4$ many zeros in the union of $I_j'$.
Given that the union has size
$$
\frac{\delta' n}{2}\cdot \left(\frac{8}{\sqrt{\delta'}}+1\right)<5\sqrt{\delta'}n,
$$
the probability is at most (where the summand $r$ is the number of zeros in the union of $I_j'$)
\[
\sum_{r=\delta'n/4}^{\delta' n} {\frac {{5 \sqrt{\delta'} n \choose r} \cdot
{n - 5 \sqrt{\delta'} n \choose \delta' n - r}}
{{n \choose \delta' n}}} \le \left(\frac{3\delta'n}{4}+1\right)\cdot {5 \sqrt{\delta'} n \choose \delta'n /4}\cdot 
{\frac { 
{n - 5 \sqrt{\delta'} n \choose 3\delta' n /4}}
{{n \choose \delta' n}}}
\]
given that the terms are maximized at $r=\delta'n/4$.
Using ${n\choose k}\le (en/k)^k$, we have 
$$
{5 \sqrt{\delta'} n \choose \delta'n /4}\le  \left(\frac{60}{\sqrt{\delta'}}\right)^{\delta'n/4}.
$$
On the other hand, we have 
$$
{\frac { 
{n - 5 \sqrt{\delta'} n \choose 3\delta' n /4}}
{{n \choose \delta' n}}}\le 
{\frac { 
{n \choose 3\delta' n /4}}
{{n \choose \delta' n}}}
=
\frac{(n-\delta'n)!\cdot (\delta'n)!}{(n-3\delta'n/4)!\cdot (3\delta'n/4)!}\le \left(\frac{\delta'n}{n-\delta'n}\right)^{\delta'n/4}\le (2\delta')^{\delta'n/4}.
$$
As a result, the probability is at most
$$
\left(\frac{3\delta'n}{4}+1\right)\cdot \left(120\sqrt{\delta'}\right)^{\delta'n/4}=2^{-\Omega(\delta'\log (1/\delta') n)}
$$
since $\delta'$ is at most some sufficiently small constant. This finishes the proof of \Cref{claim:scorelargesmallprob}.
\end{proof}

%% file: appendix.tex

\section{An upper bound on average-case zero-trace reconstruction} \label{ap:upper-bound-c2}

We recall from \Cref{sec:average-case} that in the asymptotic limit, the best possible performance of any zero-trace average-case reconstruction algorithm is given by
\[
c_2 = \lim_{n \to \infty} \max_{z \in \zo^n} {\frac {\E_{\bx \sim \zo^n}[|\LCS(\bx,\bz)|]}{n}},
\]
and from \Cref{sec:average-case-small-rho} that this quantity equals $\lim_{n \to \infty} {\frac {L_{0,\avg}(n)}{n}}.$

Via an involved analysis, Bukh and Cox show that $\E_{\bx \sim \zo^n}[|\LCS(\bx,w)|] \geq 0.82118$ where $w$ is the $n$-bit string $(0110111010010110010001011010)^{n/28},$ and hence $c_2 \geq 0.82118.$  We give an upper bound on $c_2$:

\begin{claim} \label{claim:upper-bound-c2}
$c_2 \leq 0.88999.$
\end{claim}
\begin{proof}
Fix $z \in \zo^n$ to be the optimal string that maximizes $\E_{\bx \sim \zo^n}[|\LCS(\bx,\bz)|]$. The claimed bound on $c_2$ follows from
\begin{equation}
\label{eq:mlp}
\Prx_{\bx \sim \zo^n}[z \text{~has a matching of size $0.88999n$ with $\bx$}] \leq o(1),
\end{equation}
which we establish below by showing that
\begin{equation}
\label{eq:mlp2}
\sum_{S \subseteq [n], |S| = 0.88999n}\Prx_{\bx \sim \zo^n}[z_S \text{~matches entirely into $\bx$}] \leq o(1).
\end{equation}
Via a union bound, \Cref{eq:mlp2} in turn follows from showing that for any $t$-bit string $y$, where $t:=0.88999n$, we have
\begin{equation}
\label{eq:mlp3}
\Prx_{\bx \sim \zo^n}[y \text{~matches entirely into $\bx$}] = {\frac {o(1)}{{n \choose 0.88999n}}}.
\end{equation}

Fix any $t$-bit string $y$ and any $n$-bit string $x$. 
The ``greedy strategy'' for (attempting to) entirely match $y$ into $x$ is the approach which maintains two pointers $p_y$ (into the coordinates of $y$) and $p_x$ and scans across $x$ by successively incrementing $p_x$, matching each coordinate of $y$ and incrementing $p_y$ whenever it is possible to do so.  
We recall the following well-known fact:

\begin{claim} [Greedy matching is optimal for entirely matching one string into another] \label{claim:greedy-optimal}
There is some matching that entirely matches $y$ into $x$ if and only if the greedy strategy succeeds in entirely matching $y$ into $x$.
\end{claim}


We return to establishing \Cref{eq:mlp3}.
By \Cref{claim:greedy-optimal}, $y$ matches entirely into $\bx$ if and only if the greedy strategy matches $y$ entirely into $\bx$.  
We may view a uniform $\bx \sim \zo^n$ as being generated by successively tossing coins for the successive bits of $\bx$; from this perspective it is clear that $\Pr_{\bx \sim \zo^n}[$the greedy strategy successfully matches $y$ entirely into $\bx]$ is precisely the probability that a sequence of $n$ fair coin tosses has at least $t$ ``heads'' (the $i$-th coin toss coming up ``heads'' corresponds to the $i$-th bit $\bx_i$ matching the bit of $y$ currently pointed to by $p_y$).
By \Cref{fact:standard-bound}, this probability is at most
\begin{equation} \label{eq:plum}
{\frac {2^{H(0.11001)n}}{2^n}},
\end{equation}
so again using \Cref{fact:standard-bound} and $2H(0.11001) < 1$ we get that $(\ref{eq:plum}) = {\frac {o(1)}{{n \choose 0.88999n}}}$ as required.
\end{proof}


\section{A simple upper bound on average-case one-trace reconstruction in the small deletion rate regime } \label{sec:small-delta-average-case-weak-upper-bound}

In this section we give a simple upper bound on the best possible expected LCS that any one-trace algorithm can achieve in the average-case small-deletion-rate regime. 
The argument, which is based on a union bound over all possible matchings of a given size, is significantly simpler than the proof of \Cref{thm:small-delta-average-case-upper-bound}, but it yields a result that is quantitatively weaker by a $\Theta(\log(1/\delta))$ factor.  

\begin{theorem} 
[Weak average-case upper bound on any algorithm, small deletion rate]
\label{thm:small-delta-average-case-upper-bound-weak}
Let $\delta=\delta(n)$ be any $\omega(1/n)$ deletion rate.
There is an absolute constant $c>0$ such that for sufficiently large $n$ we have
$L_{1,\avg}(\delta,n) \leq (1 - c \delta/\log(1/\delta))n.$
\end{theorem}
\begin{proof}
As in the beginning of the proof of \Cref{thm:small-delta-average-case-upper-bound}, by recalling the well-known fact \cite{KaasBuhrman80} that the median of the $\Bin(n,\delta)$ distribution belongs to $\{\lfloor n \delta \rfloor,\lceil n \delta \rceil\}$, since $\delta =\omega(1/n)$ we have that with probability $\Omega(1)$ the length $|\by|$ of a random trace $\by$ drawn from the $\delta$-deletion channel is at least $(1-\Omega(\delta))n =: (1-\delta')n$. 
Hence to upper bound $L_{1,\avg}(\delta,n)$ as claimed, it suffices to show the following: for any one-trace algorithm $A$ that is given as input a uniform random trace $\by$, of length exactly $(1-\delta')n$, from a uniform random source string $\bx \sim \zo^n$, we have 
\begin{equation} \label{eq:apxgoal}
\Prx_{\bx \sim \zo^n}\left[A \text{~outputs a hypothesis string $z$ with~}|\LCS(\bx,z)| \geq \left(1 - {\frac {c \delta'}{\log(1/\delta'}}\right)n\right] \leq 0.9.
\end{equation}

We first recall from \Cref{cor:uniform} that given a trace $\by$ of length $(1-\delta')n$ from a uniform $\bx \sim \zo^n$,  the $\delta' n$ bits of $\bx_{\bD}$ that are missing from $\by$ are independent and uniform random. 
Next, we note that any candidate matching $\mu$ of size $(1-\tau)n$ between a source string $x \in \zo^n$ and a hypothesis string $z \in \zo^n$ is completely specified by two subsets $S = \{i_1 < \cdots < i_{\tau n}\} \subset[n] $ and $S' = \{j_1 < \cdots < j_{\tau n}\}  \subset[n] $ of size $\tau n$, where $S$ ($S'$, respectively) is the set of positions in $x$ (positions in $z$, respectively) that do not participate in the matching.

Fix any hypothesis string $z \in \zo^n$ (here $z$ may depend on the trace $\by \sim \zo^{(1-\delta')n}$ that algorithm $A$ receives as input).  
Consider a fixed candidate matching $\mu$ of size $(1-\tau)n$ between $z$ and $\bx$, defined by two fixed sets $S,S'$ as described above.  
For $\tau<\delta'/2$ (which will be the case given our final parameter setting for $\tau$), even if all $\tau n$ positions in $S$ are contained in the deleted locations $\bD$, there are at least $(\delta - \tau)n \geq (\delta'/2) n$ bits in $\bx_{\bD}$ that are not present in $\by$ but are matched to some bits of $z$ by the candidate matching $\mu$. 
As mentioned above, these bits are independently uniform random, and so the probability that $\mu$ successfully matches all of those (at least) $(\delta'/2)n$ bits with the right outcomes of their partners in $z$ is at most $2^{-(\delta'/2)n}.$ 
It follows that $\Pr_{\bx}[$the candidate matching $\mu$ is a valid matching between $z$ and $\bx] \leq 2^{-(\delta'/2)n}$.  
Hence we have
\begin{align*}
&\Prx_{\bx \sim \zo^n}[\text{there exists some matching of size $(1-\tau)n$ between $\bx$ and $z$}] \\
&\leq {n \choose \tau n}^2 \cdot 2^{-(\delta'/2)n}
\leq 2^{(2H(\tau) - \delta'/2)n}
\leq 2^{-(\delta'/4)n} \leq 0.9,
\end{align*}
where the first inequality is by a union bound over all ${n \choose \tau n}^2$ many candidate matchings of size $(1-\tau)n$, the second is \Cref{fact:standard-bound}, the
third holds by choosing $\tau = c\delta'/\log(1/\delta')$ for a suitable absolute constant $c$, and the fourth (with room to spare) is because $\delta'$, like $\delta$, is $\omega(1/n).$
\end{proof}

\section{No constant-size $(2/3 + \eps)n$-LCS cover for any constant $\eps>0$} \label{ap:bestLCScover}

\begin{claim}
For any positive constant $\eps$, any $(2/3+\eps)n$-LCS cover $S\subseteq \{0,1\}^n$ 
must have size $\Omega(\log n)$.
\end{claim}
\begin{proof}
Let $\eps$ be a positive constant and let $\eps'=6\eps$.
Let $S\subseteq \zo^n$ be a $(2/3+\eps)n$-LCS cover for strings of length $n$.
As explained in \Cref{sec:bukh-ma}, by arguments given in the proof of Theorem~1.4 of \cite{GHS20}, for any $x \in \zo^n$ (and hence in particular for each string $x \in S$), there can be at most $ {1200}/{\eps'^3}$ many strings $a \in C_{n,\eps'}$ that have $\abs{\LCS(x,a)} \geq (2/3 + \eps'/6)n=(2/3+\eps)n.$ 
Say a string $a \in C_{n,\eps'}$ is \emph{covered} if there is some string $x \in S$ such that $\abs{\LCS(x,a)} \geq (2/3 + \eps)n$; it follows that at most $|S| \cdot {({1200}/{\eps'^3}})$ strings in $C_{n,\eps'}$ are covered.  
Given that every string in $C_{n,\eps'}$ is covered (by the assumption that $S$ is a $(2/3+\eps)n$-LCS cover),
  we have
$$
|S|\cdot \frac{1200}{\eps'^3}\ge \big|C_{n,\eps'}\big|=\frac{\log n}{\log (1/\eps'^4)},
$$
from which the $\Omega(\log n)$ lower bound on $|S|$ follows.
\end{proof}